\documentclass[10pt,journal,compsoc]{IEEEtran}

\ifCLASSOPTIONcompsoc
  \usepackage[nocompress]{cite}
\else
  \usepackage{cite}
\fi

\usepackage{booktabs} 
\usepackage{graphicx,amsmath,amssymb,mathrsfs,epsf}
\usepackage{tabularx}
\usepackage{stfloats}
\usepackage{setspace}
\usepackage{bbold}
\usepackage{bm}
\usepackage{subcaption}
\usepackage{bigstrut,array,multirow,tabularx}
\usepackage{tikz}
\usepackage{xcolor}
\usepackage{pgfplots}
\usepackage{dsfont}
\usepackage{bbm}
\usepackage{graphics}
\usepackage{graphicx}
\usepackage[ruled, vlined, linesnumbered]{algorithm2e}
\usepackage{slashbox}
\usepackage{makecell}
\usepackage{amsthm}

\newcommand{\N}{\mathbb{N}}
\newcommand{\R}{\mathbb{R}}

\newcommand{\argmax}{\mathop{\mathrm{arg\,max}}}

\newtheorem{definition}{Definition}
\newtheorem{theorem}{Theorem}
\newtheorem{lemma}{Lemma}

\newtheorem{corollary}{Corollary}

\newcommand\norm[1]{\left\lVert#1\right\rVert}
\SetKwRepeat{Do}{do}{while}%

\allowdisplaybreaks
\begin{document}
\raggedbottom

\title{
	\textsf{DeepNC}: Deep Generative Network Completion
}

\author {
       	 Cong Tran,~\IEEEmembership{Student Member,~IEEE}, Won-Yong Shin,~\IEEEmembership{Senior Member,~IEEE}, Andreas Spitz, \\and Michael Gertz%
\IEEEcompsocitemizethanks{\IEEEcompsocthanksitem C. Tran is with the Department of Computer Science and Engineering, Dankook University, Yongin 16890, Republic of Korea, and also with the Machine Intelligence \& Data Science Laboratory, Yonsei University, Seoul 03722, Republic of Korea.\protect\\
E-mail: congtran@ieee.org.
\IEEEcompsocthanksitem W.-Y. Shin is with the School of Mathematics and Computing (Computational Science and Engineering), Yonsei University, Seoul 03722, Republic of Korea.\protect\\
E-mail: wy.shin@yonsei.ac.kr.
\IEEEcompsocthanksitem A. Spitz is with the School of Computer and Communication Sciences, École Polytechnique Fédérale de Lausanne, Lausanne 1015, Switzerland. 
\protect\\
E-mail: andreas.spitz@epfl.ch.
\IEEEcompsocthanksitem M. Gertz is with the Institute of Computer Science, Heidelberg University, Heidelberg 69120, Germany.\protect\\
E-mail: gertz@informatik.uni-heidelberg.de.\\
(Corresponding author: Won-Yong Shin.)}}%

\markboth{IEEE Transactions on on Pattern Analysis and Machine Intelligence}
{Cong {et al.}: \textsf{DeepNC}: Network Completion Based on Deep generative models of graphs}

\IEEEtitleabstractindextext{
\begin{abstract}
Most network data are collected from partially observable networks with both missing nodes and missing edges, for example, due to limited resources and privacy settings specified by users on social media. Thus, it stands to reason that inferring the missing parts of the networks by performing \emph{network completion} should precede downstream applications. However, despite this need, the recovery of missing nodes and edges in such incomplete networks is an insufficiently explored problem due to the modeling difficulty, which is much more challenging than link prediction that only infers missing edges. In this paper, we present DeepNC, a novel method for inferring the missing parts of a network based on a \emph{deep generative} model of graphs. Specifically, our method first learns a likelihood over edges via an \emph{autoregressive generative} model, and then identifies the graph that maximizes the learned likelihood conditioned on the observable graph topology. Moreover, we propose a computationally efficient \textsf{DeepNC} algorithm that {\em consecutively} finds individual nodes that maximize the probability
in each node generation step, as well as an enhanced version using the expectation-maximization algorithm. The runtime complexities of both algorithms are shown to be {\em almost linear} in the number of nodes in the network. We empirically demonstrate the superiority of DeepNC over state-of-the-art network completion approaches.

\end{abstract}

\begin{IEEEkeywords}
Autoregressive generative model; deep generative model of graphs; inference; network completion; partially observable network
\end{IEEEkeywords}}
\maketitle

\IEEEdisplaynotcompsoctitleabstractindextext

%
\IEEEpeerreviewmaketitle

\section{Introduction}\label{sec:1}
\subsection{Background and Motivation}\label{sec:1a}
Real-world networks extracted from various biological, social, technological, and information systems tend to be only partially observable with missing both nodes and edges~\cite{kossinets2006effects}. 
For example, users and organizations may have limited access to data due to insufficient resources or a lack of authority. In social networks, a source of  incompleteness stems from privacy settings specified by users who partially or completely hide their identities and/or friendships~\cite{acquisti2015privacy}. As an example, consider a demographic analysis of Facebook users in New York City in June 2011 that showed 52.6\% of the users to be hiding their Facebook friends~\cite{dey2012facebook}. Using such incomplete network data may severely degrade the performance of downstream analyses such as community detection, link prediction, and node classification due to significantly altered estimates of structural properties (see, e.g.,~\cite{kossinets2006effects,koskinen2013bayesian,kronem,kromfac}  and references therein).

This motivates us to conduct {\em network completion} to infer the missing part (i.e., a set of both missing nodes and associated edges), prior to performing downstream applications. While intuitively similar, network completion is a much more challenging task than the well-studied {\em link prediction} and {\em low-rank matrix completion}, since it {\em jointly infers both missing nodes and edges}, while link prediction and matrix completion only infer missing edges. Although one prior contribution has attempted to address the recovery of missing nodes and edges with an algorithm, dubbed KronEM, that infers the missing parts of a graph based on the Kronecker graph model~\cite{kronem}, this current state-of-the-art model suffers from three major problems: 1) setting the size of a Kronecker generative parameter is not trivial; 2) the Kronecker graph model is inherently  designed under the assumption of a pure power-law degree distribution that not all real-world networks necessarily follow; and  3) its inference accuracy is not satisfactory.
 
As a way of further enhancing the performance of network completion, our study is intuitively motivated by the existence of {\em structurally similar} graphs with respect to graph distance, whose topologies are almost entirely observable.\footnote{In the following, we use the terms ``network" and ``graph" interchangeably.} Such similar graphs can be retrieved  from the same domain as that of the target graph (see~\cite{citeseer,protein,data_facebook} for more information). As an example, suppose that  many citizens residing in country A strongly protect the privacy of their social relationships, while citizens of country B tend to provide their friendship relations on social media. Intuitively, as long as the graph structures between two countries are similar to each other, latent information  within the (almost) complete data collected from country B can be uncovered and leveraged to infer the missing part of the collected data from country A. Additionally, the use of deep learning on graphs has been actively studied by exploiting this structural similarity of graphs (see, e.g.,~\cite{graphRNN,netgan} and references therein), which enables us to model complex structures over graphs with high accuracy. For example, the framework of recurrent neural networks (RNN) and generative adversarial networks (GAN) were recently introduced to construct deep generative models of graphs~\cite{graphRNN,netgan}. Thus, a natural question is how such {\em structural similarity} can be incorporated into the  problem of network completion by taking advantage of effective deep learning-based approaches. 

\subsection{Main Contributions}\label{sec:1b}
In this paper, we introduce \textsf{DeepNC},  a novel method for completing the missing part of an observed incomplete network $G_O$ based on a {\em deep generative} model of graphs. Specifically, we first learn a likelihood over edges (i.e., a latent representation) via an {\em autoregressive generative} model of graphs, e.g., GraphRNN~\cite{graphRNN} built upon RNN, by using a set of structurally similar graphs as training data, and then infer the missing part of the network. Unlike GraphRNN, which is only applicable to fully observable graphs, our method is capable of accommodating both observable and missing parts by imputing a number of missing nodes and edges with {\em sampled} values from a multivariate Bernoulli distribution. To this end, we formulate a new optimization problem with the aim of finding the graph that maximizes the learned likelihood conditioned on the observable graph topology. 
To efficiently solve the problem, we first propose a low-complexity \textsf{DeepNC} algorithm, termed \textsf{DeepNC-L}, that {\em consecutively} finds a single node maximizing the probability in each node generation step in a greedy fashion under the assumption that there are no missing edges between two nodes in a partially observable network $G_O$. We then present judicious approximation and computational reduction techniques to \textsf{DeepNC-L} by exploiting the {\em sparseness} of real-world networks. Second, by relaxing this assumption to deal with a more realistic scenario in which there are missing edges in $G_O$, we propose an enhanced version of \textsf{DeepNC} using the expectation-maximization (EM) algorithm, termed \textsf{DeepNC-EM}, which enables us to jointly find both missing edges between nodes in $G_O$ and edges associated with missing nodes by executing \textsf{DeepNC-L} iteratively. That is, the \textsf{DeepNC-EM} algorithm jointly solves network completion and link prediction in a single module. We show that the computational complexity of both \textsf{DeepNC} algorithms is {\em almost linear} in the number of nodes in the network. By adopting the graph edit distance (GED)~\cite{ged} as a performance metric, we empirically evaluate the performance of both \textsf{DeepNC} algorithms for various environments. Experimental results show that our   algorithms consistently outperform state-of-the-art network completion approaches by up to 68.25\% in terms of GED.  The results also demonstrate the robustness of our method not only on various real-world networks that do not necessarily follow a power-law degree distribution, but also in three more  difficult  and challenging situtations  where 1) a large portion of nodes are missing, 2) training graphs are only partially observed, and 3) a large portion of edges between  nodes  in $G_O$ are missing. Additionally, we analyze and empirically validate the computational complexity of \textsf{DeepNC} algorithms. Our main contributions are five-fold and summarized as follows: 
\begin{itemize}
\item We introduce \textsf{DeepNC}, a deep learning-based network completion method for partially observable networks; 
\item We formalize our problem as the imputation of missing data in an optimization problem that maximizes the conditional probability of a generated node sequence;
\item We design two computationally efficient \textsf{DeepNC} algorithms to solve the problem  by exploiting the sparsity of networks;\footnote{The source code used in this paper is available online (https://github.com/congasix/DeepNC).} 
\item We validate \textsf{DeepNC} through extensive experiments using real-world datasets  across various domains, as well as synthetic datasets;
\item We analyze and empirically validate the computational complexity of \textsf{DeepNC}.
\end{itemize}
To the best of our knowledge, this study is the first work that applies deep learning to network completion.

\subsection{Organization and Notations}\label{sec:1c}
The remainder of this paper is organized as follows. In Section~\ref{sec:2}, we summarize significant studies that are related to our work. In Section~\ref{sec:3}, we explain the methodology of our work, including the problem definition and an overview of our \textsf{DeepNC} method. Section~\ref{sec:4}  describes implementation details of the two \textsf{DeepNC} algorithms and analyzes their computational complexities. Experimental results are discussed in Section~\ref{sec:5}. Finally, we provide a summary and concluding remarks in Section~\ref{sec:6}. 

Table~\ref{tab:notation} summarizes the notation that is used in this paper. This notation will be formally defined in the following sections when we introduce our methodology and the technical details.

\begin{table}[t!]
    \centering
    \caption{Summary of notations}
    \label{tab:notation}

    \begin{tabular}{|l|p{6cm}|}
        \hline
        \multicolumn{1}{|c|}{\textbf{Notation}}&\multicolumn{1}{c|}{\textbf{Description}}                                                                                                \\ \hline
$G_T$ & true graph      \\ \hline
$G_O$ & partially observable graph      \\ \hline
$V_O$ & set of nodes in $G_O$      \\ \hline
$E_O$ & set of edges in $G_O$      \\ \hline
$V_M$ & set of missing nodes      \\ \hline
$E_M$ & set of missing edges      \\ \hline
$G_I$ & training graph      \\ \hline
$p_\text{model}$ & probability distribution over edges of a graph      \\ \hline
$\Theta$ & learned parameters of $p_\text{model}$ \\ \hline
$\hat{G}$ & recovered graph      \\ \hline
$\pi$ & node ordering      \\ \hline
${\bf S}^\pi$ & a sequence of nodes and edges under $\pi$\\ \hline
\end{tabular}

\end{table}

\section{Related Work}\label{sec:2}

The method that we propose in this paper is related to four broader areas of research, namely generative models of graphs, link prediction, low-rank matrix completion, and network completion.

\textbf{Generative models of graphs.} The study of generative models of graphs has a long history, beginning with the first random model of graphs that robustly assigns probabilities to large classes of graphs, and was introduced by Erd\H{o}s and R{\'e}nyi~\cite{data_com}. Another well-known model generates new nodes based on preferential attachment~\cite{data_ba}. More recently, a generative model based on Kronecker graphs, the so-called KronFit~\cite{kronfit}, was introduced, which generates synthetic networks that match many of the structural properties of real-world networks such as constant and shrinking diameters. Recent advances in {\em deep learning}-based approaches have made further progress towards generative models for complex networks~\cite{graphRNN, grans, netgan, graphvae, graphvae2, gcpn, miscgan, GN}. GraphRNN~\cite{graphRNN} and graph recurrent attention networks (GRAN)~\cite{grans} were presented to learn a distribution over edges by decomposing the graph generation process into sequences of node and edge formations via autoregressive generative models; an approach using the Wasserstein GAN objective in the training process was applied to generate discrete output samples~\cite{netgan}; variational autoencoders (VAEs) were employed to design another deep learning-based generative model of graphs~\cite{graphvae,graphvae2}; a graph convolutional policy network was presented for goal-directed graph generation (e.g., drug molecules) using reinforcement learning~\cite{gcpn}; a multi-scale graph generative model, named Misc-GAN, was introduced by modeling the underlying distribution of graph structures at different levels of granularity to aim at generating graphs having similar community structures~\cite{miscgan}; and a more general deep generative model was presented to learn distributions over any arbitrary graph via graph neural networks~\cite{GN}. Among the aforementioned methods, autoregressive generative models such as GraphRNN and GRAN are the most scalable and flexible approaches in terms of graph size, while others are beneficial in generating non-topological information such as node attributes. Table~\ref{tab:literature} summarizes the literature overview of the aforementioned deep generative models of graphs.

\begin{table*}[t]
\centering
\caption{Summary of deep generative models of graphs}
\label{tab:literature}
\centering
\begin{tabular}{l c c c }
\hline
Deep generative models of graphs                                             & Scalable      & Flexible  &   Attributed  \\ \hline
Autoregressive \cite{graphRNN,grans}  &  \checkmark      & \checkmark                  &                                                     \\ 
GAN \cite{netgan,miscgan}                  & \checkmark  &  &                                 \\
VAE \cite{graphvae,graphvae2}                  &    &  & \checkmark                                   \\ 
Reinforcement learning \cite{gcpn}                     &    &  \checkmark    &  \checkmark     \\ 
General neural network \cite{GN}                       &                  &  \checkmark &  \checkmark   \\ \hline
\end{tabular}
\end{table*}

{\bf Link prediction.} Inferring the presence of links in a given network according to the neighborhood similarity of existing connections is a longstanding task in network science. Although numerous algorithms have been developed based on traditional statistical measures~\cite{lp1} and deep learning such as graph neural networks~\cite{graphvae2,linkpred2}, existing link prediction methods are not inherently designed to solve the network completion problem that jointly recovers missing nodes and edges in partially observable networks. Specifically, when a node is completely missing from the underlying network, link prediction models can no longer exploit structural neighborhood information.

{\bf Low-rank matrix completion.} Missing entries in a low-rank matrix due to partial observations can be inferred by solving
the rank minimization problem using approximation methods such as singular value decomposition~\cite{SVDlowrank}, matrix factorization~\cite{MFlowrank}, neural networks~\cite{CNNlowrank}, and adaptive clustering of bandit strategies~\cite{onlinebandit1,onlinebandit2}. Since many graphs tend to exhibit low-rank connectivity structures~\cite{lowrank}, several techniques used in matrix completion can also be applied to perform link prediction~\cite{linkMF}. Similarly as in the setting of link prediction, low-rank matrix completion requires at least one entry in each row and column to be known in order to infer missing entries.

{\bf Network completion.} Observing a partial sample of a network and inferring the remainder of the network is referred to as {\em network completion}. As the most influential prior work, KronEM, an approach based on Kronecker graphs to solving the network completion problem by applying the EM algorithm, was suggested by Kim and Leskovec \cite{kronem}. MISC was developed  to tackle the {\em missing node identification} problem when the information of connections between missing nodes and observable nodes is assumed to be available~\cite{misc}. A follow-up study of MISC~\cite{sami} incorporated metadata such as demographic information and the nodes’ historical behavior into the inference process. Furthermore, a graph upscaling method, termed EvoGraph~\cite{evograph}, can be regarded as a network completion method using a preferential attachment mechanism. 

{\bf Discussion.} Despite these contributions, no prior work in the literature exploits the power of deep generative models in the context of network completion. Although generative models of graphs such as GraphRNN can be used as a network completion method, nontrivial extra tasks are required, including computationally expensive {\em graph matching} to find the correspondence between generated graphs and the partially observable network. Furthermore, MISC and other follow-up studies do not truly address network completion,
since they solve the {\em node identification} problem under the assumption that the connections between missing nodes and observable nodes are known beforehand, which is not feasible in a setting where only partial observation of nodes is possible as we address with \textsf{DeepNC} in the following.

\section{Methodology}\label{sec:3}

As a basis for the proposed \textsf{DeepNC} algorithm in Section 4, we first describe our network model with basic assumptions and formalize the problem. Then, we explain  a deep  generative graph model and our research methodology adopting the  deep generative graph model to solve the problem of network completion.
\subsection{Problem Definition}
\label{sec:3a}
\subsubsection{Network Model and Basic Assumptions}
\label{sec:3a1}
Let us denote a partially observable network as $G_O = (V_O,E_O)$, where $V_O$ and $E_O$ are the set of vertices and the set of edges, respectively. The network $G_O$ with $|V_O|$ observable nodes can be interpreted as a subgraph taken from an underlying true network $G_T = (V_O~\cup~V_M,~E_O~\cup~E_M)$, where $V_M$ is the set of unobservable (missing) nodes and $E_M$ is the set of three types of unobservable (missing) edges, including i) the edges connecting two nodes in $V_M$, ii) the edges connecting one node in $V_O$ and another node in $V_M$, and iii) the missing edges connecting two nodes in $V_O$.  More specifically, the set of observable edges, $E_O$, is regarded as a subset of all true edges connecting nodes in $V_O$. In contrast to the conventional setting that assumes no missing edges between two nodes in $V_O$~\cite{kronem}, we relax this assumption by not requiring that $G_O$ is a complete subgraph. In the following, we assume both $G_O$ and $G_T$ to be {\em undirected unweighted} networks without self-loops or repeated edges.

Let us denote $p_\text{model}$ as a family of probability distributions over the edges of a graph, which can be parameterized by a set of model parameters $\Theta$, i.e., $({p}_\text{model}^\Theta; \Theta)$.\footnote{To simplify notations, $p_\text{model}^{\Theta}$ will be written as $p_\text{model}$ if omitting $\Theta$ does not cause any confusion.} In this paper, we assume that $G_T$ is a sample drawn from the distribution $p_\text{model}$. Furthermore, we assume that the number of missing nodes, $|V_M|$, is available or can be estimated. In practice, $|V_M|$ can be readily estimated by standard statistical methods; for example,  a latent non-random mixing model  in~\cite{estimatepopulation} is capable of estimating a network size $|V_O \cup V_M|$ by asking respondents how many people they know in specific subpopulations. For an overview of network-relevant notations, see Fig.~\ref{fig:overall}.

\subsubsection{Problem Formulation}
\begin{figure}[t]
    \begin{center}
            \includegraphics[width=\columnwidth]{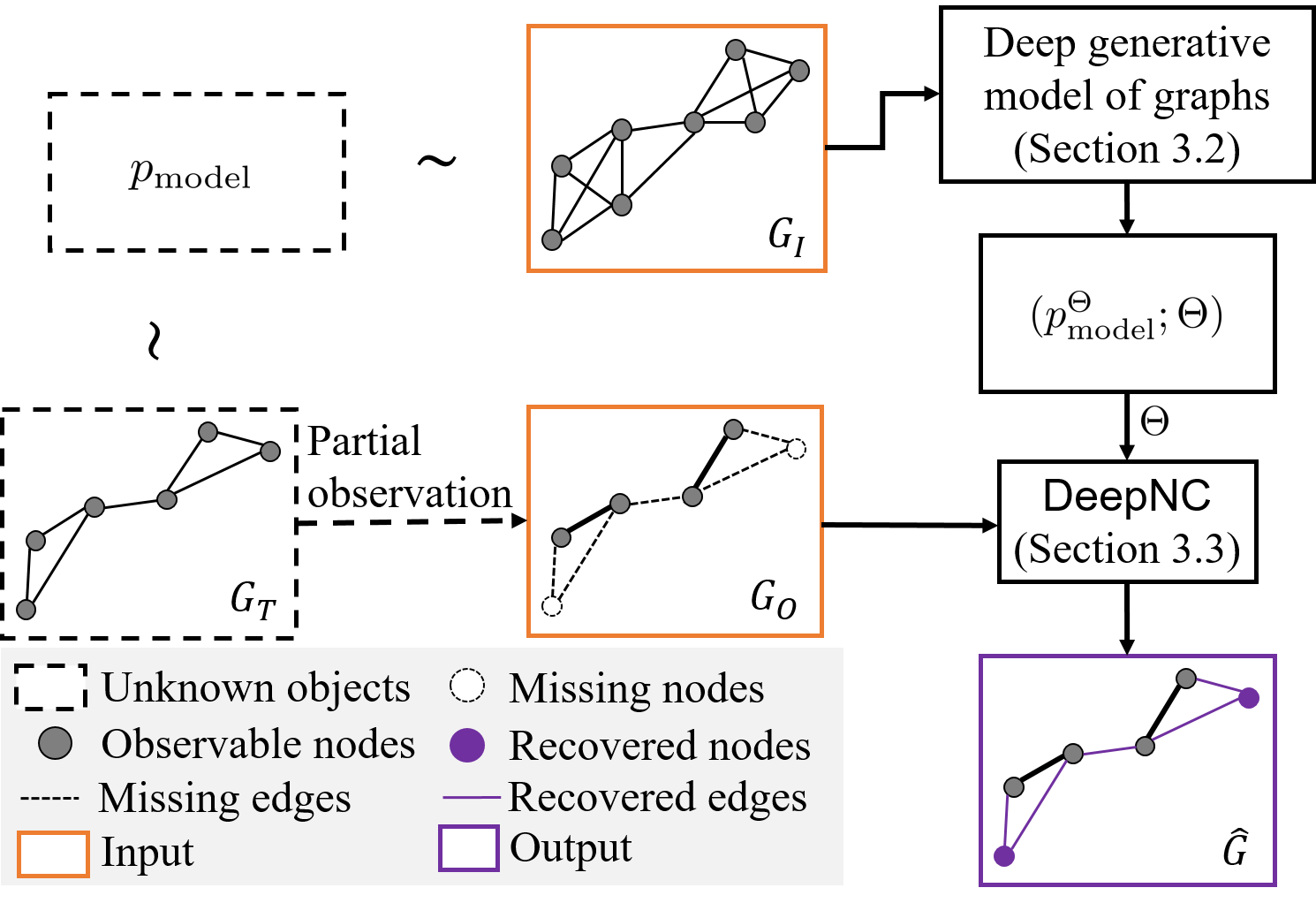}
            \caption{The schematic overview of our \textsf{DeepNC} method.}
            \label{fig:overall}
    \vspace{-0.4cm}
    \end{center}
\end{figure}

In the following, we formally define the network completion problem, the idea behind our approach, and the problem formulation. 

\begin{definition}
{\upshape \textbf{Network completion problem.}} Given a partially observable network $G_O$, network completion aims to recover all missing edges connecting nodes in the true network $G_T$ so that the inferred network, denoted by $\hat{G}$, is equivalent to $G_T$ (up to isomorphism).
\end{definition}
As illustrated in Fig.~\ref{fig:overall}, a network $\hat{G}$ is inferred using the partially observable network $G_O$  as input of \textsf{DeepNC}. We tackle this problem by minimizing a distance metric $\delta(G_T,\hat{G})$ that measures the difference between $G_T$ and $\hat{G}$.
Due to the fact that the true network $G_T$ is not available, our intuition is to analyze the connectivity patterns of one (or multiple) fully observed network(s) $G_I$ whose structure is similar to that of $G_T$ (i.e., $\delta(G_T,G_I)$ is sufficiently small) and then to make use of this information for recovering the network $G_O$, where $G_I$ is a sample drawn from the distribution $p_\text{model}$.\footnote{The number of nodes in $G_I$ should be greater than or equal to that in $G_T$ so that the information  (i.e., the distribution $p_\text{model}$) encoded by learned parameters $\Theta$ is sufficient to infer $G_T$.} To this end, we first learn $({p}_\text{model}; \Theta)$ by using $G_I$ as the {\em training} data under a deep generative model of graphs described in Section~\ref{sec:3b}. Afterwards, we generate graphs with similar structures via the set of learned model parameters $\Theta$. Among all generated graphs $G \in \mathcal{G}$ having $|V_O|+|V_M|$ nodes and containing a subgraph isomorphic to $G_O$, we find the most likely graph configuration $\hat{G}$ from the distribution over graphs in the set $\mathcal{G}$ parametrized by $\Theta$. In this context, our optimization problem can be formulated as follows:
\begin{equation}
\label{eq:probform}
\begin{aligned}
&\hat{G} = \argmax_{G \in\mathcal{G}} p(G|G_O;\Theta) \\
&\text{s.t. } |V_G| = |V_O|+|V_M|,
\end{aligned}
\end{equation}
where $|V_G|$ denotes the number of nodes in $G$.
The overall procedure of our approach is visualized in Fig.~\ref{fig:overall}. 

\subsection{Deep Generative Model of Graphs}\label{sec:3b}

Deep generative models of graphs have the ability to approximate any distribution of graphs with minimal assumptions about their structures~\cite{graphRNN,GN}. Among recently introduced deep generative models, we adopt GraphRNN~\cite{graphRNN} in our study due to its state-of-the-art performance in generating diverse graphs that match the structural characteristics of a target set as well as the scalability to much larger graphs than those from other deep generative  models (refer to Section 4 and Corollary 1 in~\cite{graphRNN} for more details). In this subsection, we briefly describe a variant of GraphRNN, termed simplified GraphRNN (GraphRNN-S), where the probability of edge connections for a node is assumed to be independent of each other. This method effectively learns $({p}_\text{model}; \Theta)$ from the set of {\em structurally similar} network(s) $G_I$. 

We first describe how to vectorize a graph. Given a graph $G$ sampled from  the distribution $p_\text{model}$ with a number of nodes equal to $|V_O|+|V_M|$, we define a node ordering $\pi$ that maps nodes to rows or columns of a given adjacency matrix of $G$ as a permutation function over the set of nodes. Thus, \{$\pi(v_1),\cdots,\pi(v_{|V_O|+|V_M|})$\} is a permutation of \{$v_1,\cdots,v_{|V_O|+|V_M|}$\}, yielding $(|V_O|+|V_M|)!$ possible node permutations. 
Given a node ordering $\pi$, a sequence ${\bf{S}}$ is then defined as:
\begin{equation}
\label{eq:1}
{\bf{S}}^\pi \triangleq ({\bf{S}}_1^\pi, \cdots, {\bf{S}}_{|V_O|+|V_M|}^\pi),
\end{equation}
where each element ${\bf{S}}^\pi_i \in \{0,1\}^{i-1}$ for $i \in \{2,\cdots,|V_O|+|V_M|\}$ is a binary adjacency vector representing the edges between node $\pi(v_i)$ and the previous nodes $\pi(v_j)$  for $j \in \{1,\cdots,i-1\} $ that already exist in the graph, and ${\bf{S}}^\pi_1 = \varnothing$. Here, ${\bf S}_i^\pi$ can be expressed as
\begin{equation}
\label{eq:2}
{\bf{S}}^{\pi}_i = ({a}^\pi_{1,i}, \cdots,{a}^\pi_{i-1,i}),\quad \forall i \in \{2,\cdots,|V_O|+|V_M|\},
\end{equation}
where ${a}_{u,v}^\pi$ denotes the $(u,v)$-th element of the adjacency matrix ${\bf {A}}^{\pi} \in \{0,1\}^{(|V_O|+|V_M|) \times (|V_O|+|V_M|)}$ for $u, v\in\{1,\cdots,|V_O|+|V_M|\}$ (refer to Fig.~\ref{fig:grnnmethodology} for an illustration of the sequence). Due to the fact that the graphs are discrete objects, the graph generation process involves discrete decisions that are not differentiable and therefore problematic for backpropagation. Thus, instead of directly learning the distribution $p_\text{model}$, we sample $\pi$ from the set of $(|V_O|+|V_M|)!$ node permutations to generate the sequences ${\bf {S}}^\pi$ and learn the distribution $p({\bf {S}}^\pi)$ over sequences.

Next, we explain how to characterize the distribution $p({\bf {S}}^\pi)$. Due to the sequential nature of ${\bf{S}}^\pi$, the distribution $p({\bf{S}}^\pi)$  for a given node ordering $\pi$ can be decomposed into the product of conditional probability distributions over the elements as follows:
\begin{equation}
\label{eq:3}
p({\bf{S}}^{\pi}) = \prod^{|V_O|+|V_M|}_{i=2}p({\bf{S}}^\pi_i|{\bf{S}}^\pi_1,\cdots,{\bf{S}}^\pi_{i-1}).
\end{equation}
For ease of notation, we simplify $p({\bf{S}}^\pi_i|{\bf{S}}^\pi_1,\cdots,{\bf{S}}^\pi_{i-1})$ as $p({\bf{S}}^\pi_i|{\bf{S}}^\pi_{<i})$ for the remainder of the paper. 

\begin{figure}[t]
    \begin{center}
            \includegraphics[width=9cm]{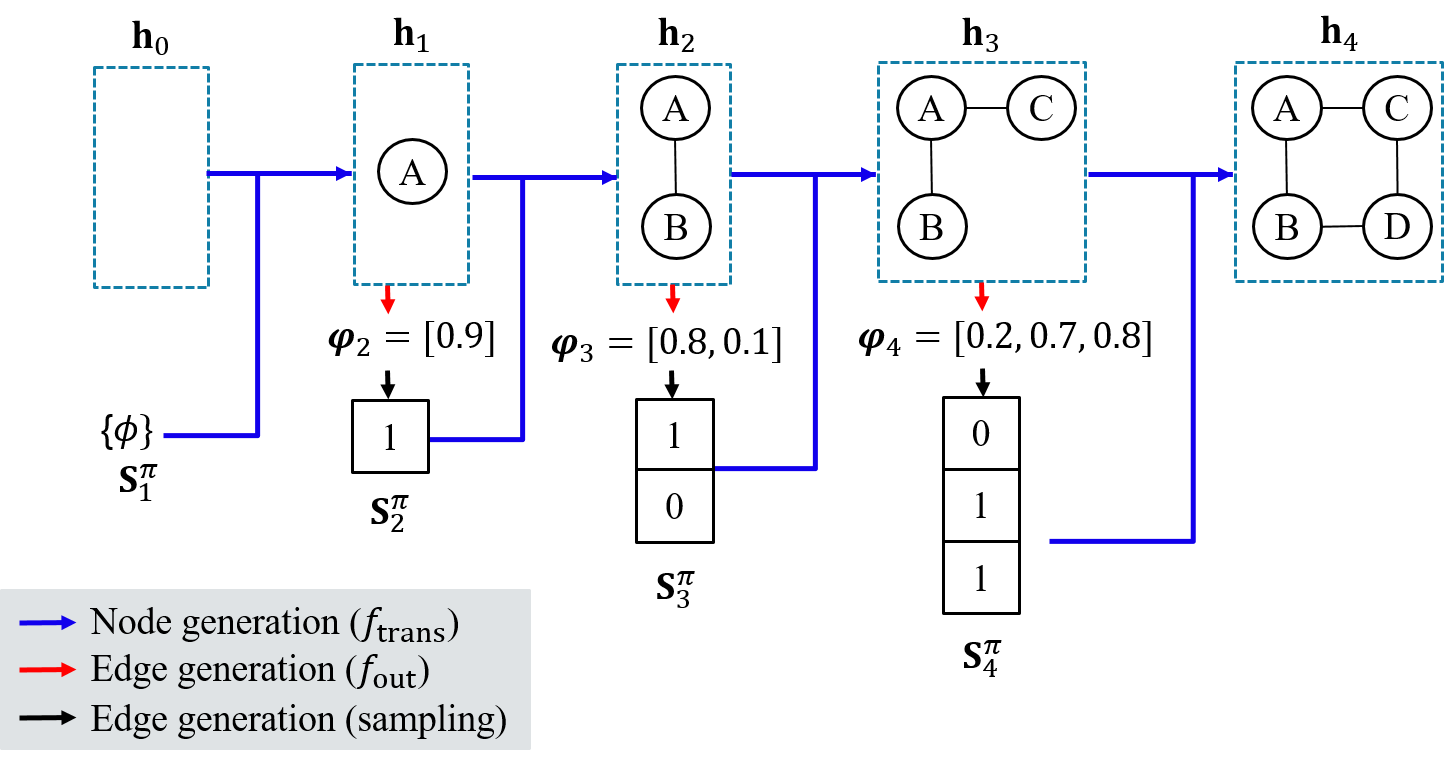}
            \caption{An example illustrating the inference process of GraphRNN-S. Here, the blue arrows denote the graph-level RNN that encodes the ``graph state'' vector ${\bf h}_i$ in its hidden state, and the red and black arrows represent the edge generation process whose input is given by the graph-level RNN.}
            \label{fig:grnnmethodology}
    \vspace{-0.4cm}
    \end{center}
\end{figure}

Now, we turn to describing the use of RNN in generating a sequence ${\bf {S}}^\pi$ from the training data $G_I$. The core idea is 
to learn two functions $f_\text{trans}$ and $f_\text{out}$ that are used in each generation step according to the following procedure (refer to Fig.~\ref{fig:grnnmethodology}).  We denote ${\bf h}_i \in \R^d$ as the graph state vector representing the hidden state of the model in the $i$-th step, where $d \in \N$ is a user-defined parameter that is typically set to a value smaller than $|V_O|+|V_M|$. A state-transition function $f_\text{trans}$ is used to compute the graph state vector ${\bf h}_i$ based on both the previous hidden state ${\bf h}_{i-1}$ and the input ${\bf{S}}^\pi_{i}$, and is given by
\begin{equation}
\label{eq:5a}
{\bf h}_i = f_\text{trans}({\bf h}_{i-1}, {\bf{S}}^\pi_i).
\end{equation}
Intuitively, ${\bf h}_i$ encodes the topological information of $i$ generated nodes in a low-dimensional vector. For the first generation step, we randomly initialize ${\bf h}_0$ and set ${\bf{S}}^\pi_1 = \varnothing$ to produce ${\bf h}_1$. Then, as the output of the $i$-th step of GraphRNN-S,  an output function $f_\text{out}$ is invoked to obtain a vector ${\bm \varphi}_{i+1} \in (0,1)^{i}$  specifying the distribution of the next node's adjacency vector as
\begin{equation}
\label{eq:5}
{\bm \varphi}_{i+1} = f_\text{out}({\bf h}_{i}). 
\end{equation}
In GraphRNN-S, $p({\bf{S}}^\pi_i|{\bf{S}}^\pi_{<i})$ is modeled as a multivariate Bernoulli distribution parametrized by ${\bm \varphi_i}$. Thus, every entry of ${\bm \varphi_i}$ in (\ref{eq:5}) can be interpreted as a probability representing whether there exists an edge between nodes $i$ and $j$ for $j \in \{1,\cdots,i-1\}$. The function $f_\text{trans}$ is found via general neural networks such as gated recurrent units (GRUs)~\cite{GRU} or long short-term memory (LSTM) units~\cite{LSTM} in RNN, and the function $f_\text{out}$ is a multilayer perceptron. The weights of  $f_\text{trans}$ and $f_\text{out}$ are optimized using training sequences sampled from $G_I$ (refer to~\cite{graphRNN} for further details on the training process). It is worth noting that, rather than learning to generate graphs under any possible node permutations, GraphRNN-S learns from samples generated via breadth-first search (BFS) to allow the training process to be tractable.

A set of model parameters $\Theta$ is referred to as learned weights of both $f_\text{trans}$ and $f_\text{out}$ after the training process. Fig.~\ref{fig:grnnmethodology} illustrates the inference process of GraphRNN-S, where a graph consisting of four nodes is generated as depicted from left to right. In more detail, after obtaining $\varphi_2$ via (\ref{eq:5a}) and (\ref{eq:5}), ${\bf{S}}^\pi_2 = [1]$ is acquired by sampling from the multivariate Bernoulli distribution parameterized by ${\bm \varphi}_2$, which means that the next generated node (i.e., node B) is linked to node A. Following a similar procedure, we obtain ${\bf{S}}^\pi_3 = [1,0]$ and ${\bf{S}}^\pi_4 = [0,1,1]$ representing the connections of nodes C and D with previously generated nodes, respectively.

\begin{figure*}[th]
    \begin{center}
            \includegraphics[width=0.8\linewidth]{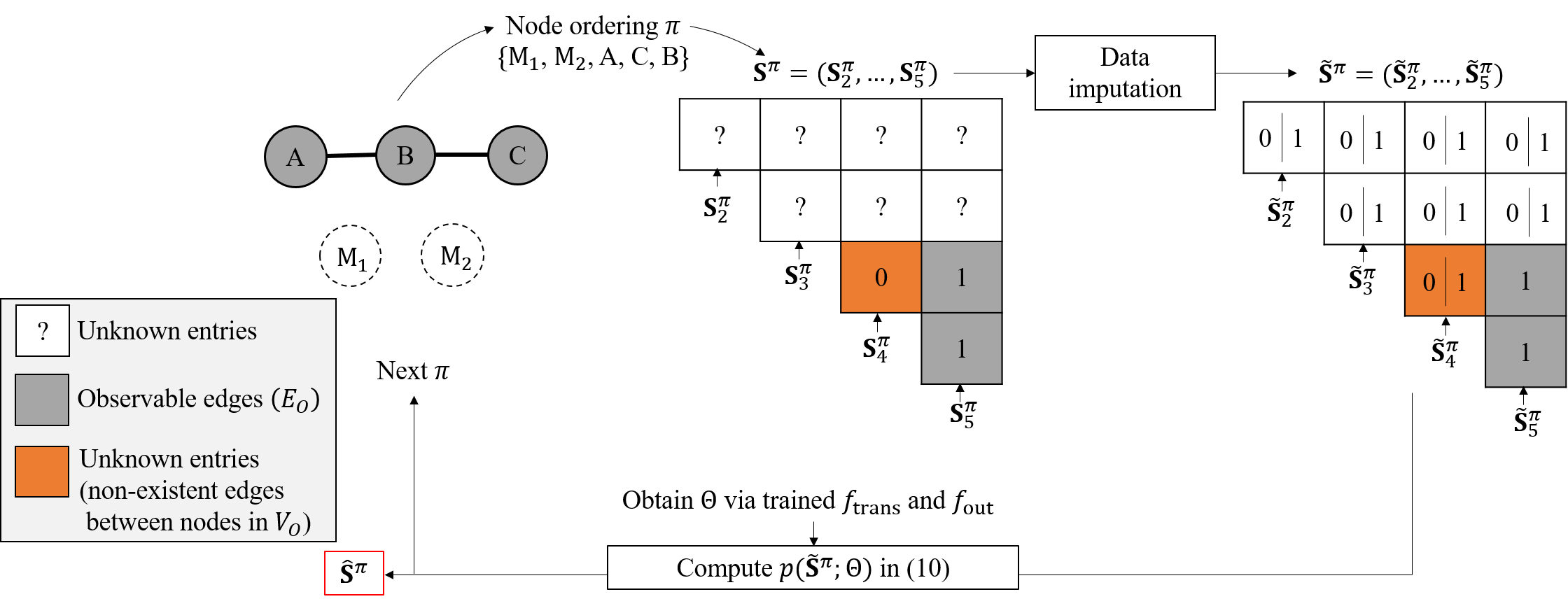}

            \caption{An example illustrating the schematic overview of our \textsf{DeepNC} method, where three nodes (i.e., A, B, and C)  and two edges with solid lines are observable instead of the true graph ${G}_T$ consisting of five nodes and all associated edges. Both white and orange entries in ${\bf S}^\pi$ are imputed with either 0 or 1 while grey entries  in ${\bf S}^\pi$ remain unchanged.}
            \label{fig:methodology}
    \end{center}
\end{figure*}

\subsection{Network Completion} \label{sec:3c}

In this subsection, we present our \textsf{DeepNC} method that recovers the missing part of the true network $G_T$ based on the deep generative model. We first  describe the approach that seamlessly accommodates both observable and missing parts of $G_T$ into the graph generation process using the  trained functions $f_\text{trans}$ and $f_\text{out}$ in Section~\ref{sec:3b}. Then, we present the problem reformulation built upon (\ref{eq:probform}).

We start by modeling the graphs that we want to recover as sequences and incorporating the information from the observed graph $G_O$ into the generation process. To this end, we reuse the notation ${\bf{S}}^\pi$ in (\ref{eq:1}) so that the sequence accounts for both observable and missing parts, where indices of missing nodes correspond to placeholders (e.g., M$_1$ and M$_2$ in Fig.~\ref{fig:methodology}), if such inclusion of unknown entries in ${\bf S}^\pi$ does not cause any confusion.
Then, we solve (\ref{eq:probform}) through {\em data imputation} of the unknown entries (i.e., the entries associated with missing nodes), which also include non-existent edges between nodes in $G_O$. Let  
\begin{align}
\label{eq:5b}
{\bf \tilde{S}}^\pi=(\tilde{\bf S}_1^\pi,\cdots,\tilde{\bf S}_{|V_O|+|V_M|}^\pi) 
\end{align} 
denote the sequence that we obtain from data imputation under a node ordering $\pi$, which contains both observable edges taken directly from ${\bf{S}}^\pi$, corresponding to the set $E_O$, and possible instances of all missing entries.
Then, we impute each missing entry in ${\bf S}^\pi$ with either 0 or 1, thereby yielding $2^{\frac{(|V_O|+|V_M|)(|V_O|+|V_M|-1)}{2}-|E_O|}$ possible outcomes of ${\bf \tilde{S}}^\pi$, where data imputation for non-existent edges between nodes in $G_O$  (i.e., orange entries in ${\bf S}^\pi$ of Fig.~\ref{fig:methodology}) can be thought of as {\em link prediction} since structural neighborhood information regarding observable nodes is available. For each outcome, we use trained functions $f_\text{trans}$ and $f_\text{out}$ to obtain the corresponding ${\bm \varphi_i}$ for $i\in\{2,\cdots,|V_O|+|V_M|\}$. 

Next, since the constraint in (\ref{eq:probform}) is incorporated into ${\bf \tilde{S}}^\pi$ in (\ref{eq:5b}), we reformulate our optimization problem in (\ref{eq:probform}) as finding a sequence $\hat{\bf{S}}^{{\pi}}$ that maximizes the probability $p({\bf \tilde{S}}^\pi;\Theta)$ under a node ordering ${\pi}$ from a distribution of sequences parametrized by $\Theta$ as follows:
\begin{equation}
\label{eq:6}
\hat{\bf{S}}^{{\pi}}=\argmax_{{\bf \tilde{S}}^\pi}p({\bf \tilde{S}}^\pi;\Theta), \\
\end{equation} 
which can be computed as 
\begin{equation}
\label{eq:4b}
\begin{aligned}
p({\bf \tilde{S}}^\pi;\Theta) 
& = p({\bf \tilde{S}}^\pi_2|{\bf \tilde{S}}^\pi_{<2};\Theta)p({\bf \tilde{S}}^\pi_3|{\bf \tilde{S}}^\pi_{<3};\Theta)\cdots \\
& \quad\:\, p({\bf \tilde{S}}^\pi_{|V_O|+|V_M|}|{\bf \tilde{S}}^\pi_{<|V_O|+|V_M|};\Theta) \\
& = p({\bf \tilde{S}}^\pi_2;\{{\bf h}_1,{\bm \varphi_2}\})p({\bf \tilde{S}}^\pi_3;\{{\bf h}_2,{\bm \varphi_3}\})\cdots \\
& \quad\:\, p({\bf \tilde{S}}^\pi_{|V_O|+|V_M|};\{{\bf h}_{|V_O|+|V_M|-1},{\bm \varphi_{|V_O|+|V_M|}}\})\\
& = p({\bf \tilde{S}}^\pi_2;{\bm \varphi_2})p({\bf \tilde{S}}^\pi_3;{\bm \varphi_3})\cdots p({\bf \tilde{S}}^\pi_{|V_O|+|V_M|};{\bm \varphi_{|V_O|+|V_M|}})\\
& = \prod^{|V_O|+|V_M|}_{i=2}p({\bf \tilde{S}}^\pi_i;{\bm \varphi_i}),
\end{aligned}
\end{equation}
where  the first equality follows due to (\ref{eq:3}); the second equality holds since $\Theta$ is the set of learned model parameters of both $f_\text{trans}$ and $f_\text{out}$ in (\ref{eq:5a}) and (\ref{eq:5}), respectively; and the third equality stems from the fact that ${\bf \tilde{S}}^\pi_i$ is determined only by ${\bm \varphi_i}$.
Since ${\bm \varphi}_i$ is the set of variables of a multivariate Bernoulli distribution  in which each entry  represents the likelihood of edge existence, we have:
\begin{equation}
\label{eq:4c}
p({\bf \tilde{S}}^\pi;\Theta) = \prod_{i=2}^{|V_O|+|V_M|}\left(\prod_{\tilde{s}^\pi_{i,j} = 1}{\bm \varphi}_{i,j}\prod_{\tilde{s}^\pi_{i,j} = 0}(1-{\bm \varphi}_{i,j})\right),
\end{equation}
where $\tilde{s}^\pi_{i,j}$ denotes the $j$-th element of the binary vector ${\bf \tilde{S}}^\pi_i$ for $i\in\{2,\cdots,|V_O|+|V_M|\}$ and $j\in\{1,\cdots,i-1\}$; and ${\bm \varphi}_{i,j} \in (0,1)$ is the $j$-th element of ${\bm \varphi_i}$.
An example visualizing our \textsf{DeepNC} method is presented in Fig.~\ref{fig:methodology}, where we observe a network $G_O$  consisting of three nodes (i.e., A, B, and C) and two edges, instead of the true network ${G}_T$ with 5 nodes (i.e., A, B, C, M$_1$, and M$_2$). 

To solve (\ref{eq:6}), we need to compute $p({\bf \tilde{S}}^\pi;\Theta)$ via exhaustive search over $(|V_O|+|V_M|)!$ node permutations. 
Since computing $p(\tilde{\bf S}^\pi;\Theta)$  in (\ref{eq:4b}) requires $\frac{(|V_O|+|V_M|)^2}{2}$ multiplication operations and data imputation yields $2^{\frac{(|V_O|+|V_M|)(|V_O|+|V_M|-1)}{2}-|E_O|}$ possible outcomes of $\bf \tilde{S}^\pi$, its computational complexity is bounded by $\mathcal{O}((|V_O|+|V_M|)^22^{\frac{(|V_O|+|V_M|)(|V_O|+|V_M|-1)}{2}-|E_O|}(|V_O|+|V_M|)!)$. This motivates us to introduce a low-complexity algorithm in the next section for efficiently solving this problem.

\section{\textsf{DeepNC} Algorithms}\label{sec:4}

In this section, we introduce two algorithms that we design to efficiently solve the network completion problem in (\ref{eq:6}). In designing such algorithms, we focus on how to compute the likelihood of edge existence in the form of a tuple $(\hat{\pi},\Phi)$, where $\hat{\pi}$ represents a node ordering to be inferred and $\Phi=\{{\bm \varphi}_2,\cdots, {\bm \varphi}_{|V_O|+|V_M|}\}$. Then, $\bf \hat{S}^\pi$ in (\ref{eq:6}) can be acquired by sampling from $(\hat{\pi},\Phi)$. First, we present \textsf{DeepNC-L}, a \underline{l}ow-complexity \underline{deep} \underline{n}etwork \underline{c}ompletion algorithm, working based on the assumption that a partially observable graph $G_O$ is a complete subgraph with no missing edges. Second, we present an enhanced version of \textsf{DeepNC-L} using the EM algorithm~\cite{emalgorithm}, dubbed \textsf{DeepNC-EM}, to deal with the case where edges are missing in $G_O$. The overall architecture of both \textsf{DeepNC} algorithms is illustrated in Fig.~\ref{fig:algo_overall}. We also analyze  their computational complexities. 

\begin{figure}[t]
    \begin{center}
            \includegraphics[width=0.65\linewidth]{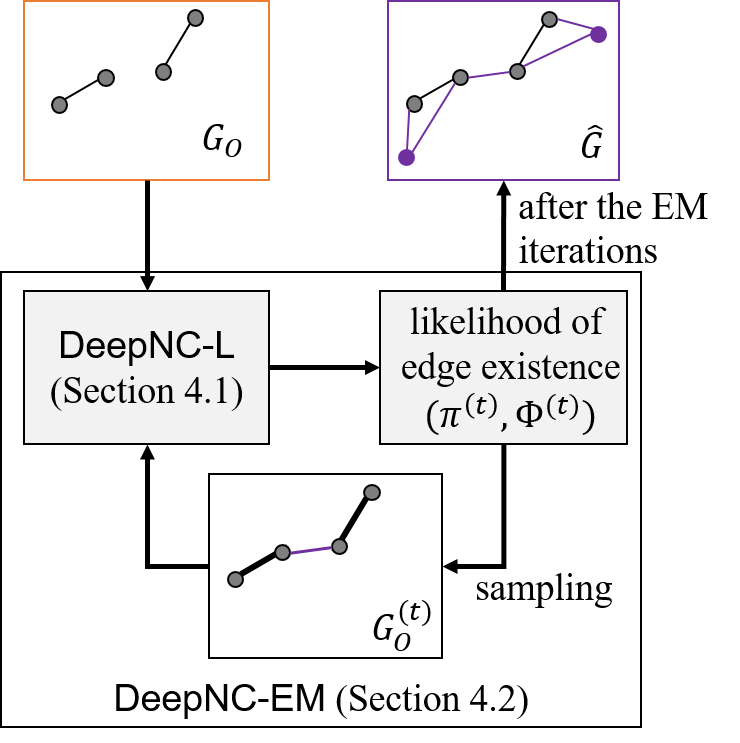}
            \caption{The overall architecture of \textsf{DeepNC} algorithms.}
            \label{fig:algo_overall}
    \end{center}
\end{figure}

\subsection{\textsf{DeepNC-L} Algorithm}\label{sec:4a}

\subsubsection{Overall Procedure}

\begin{figure*}[t]
    \begin{center}
            \includegraphics[width=0.75\linewidth]{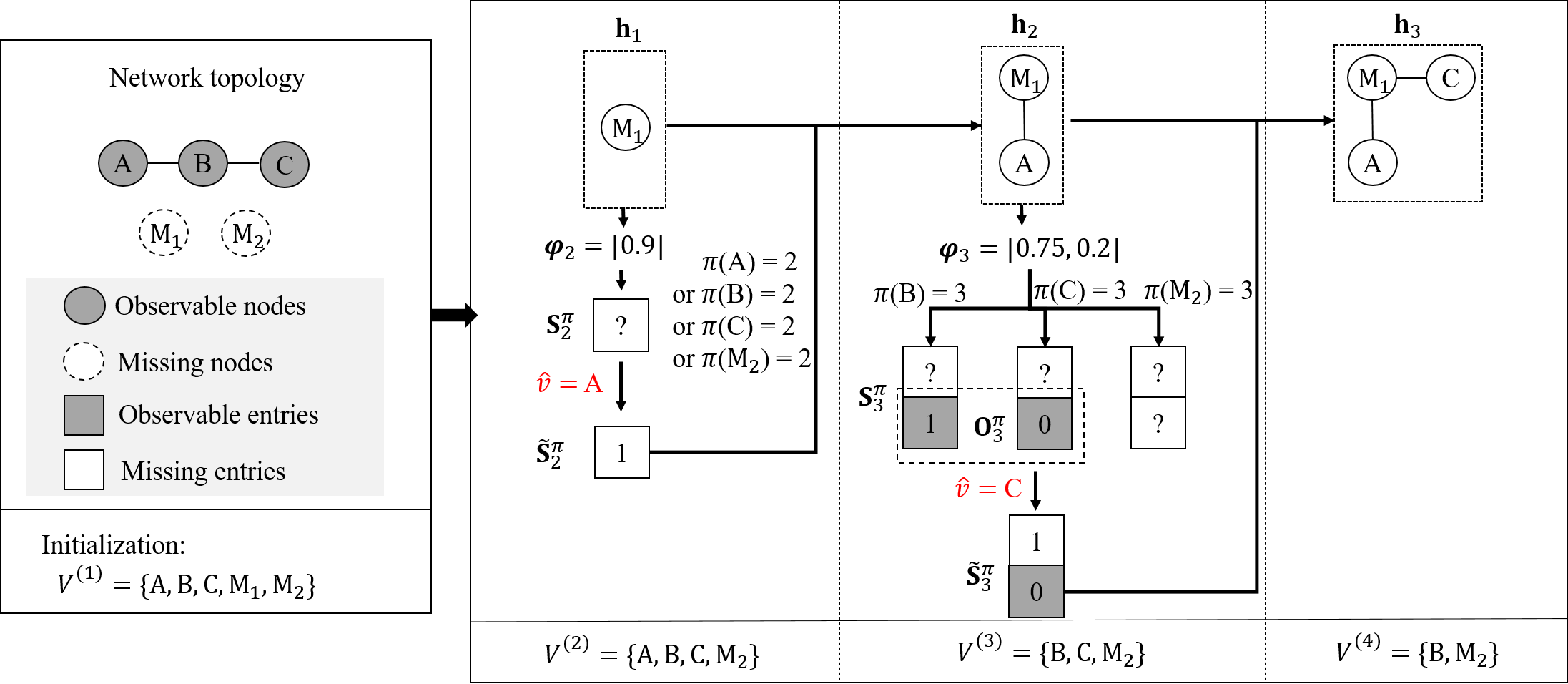}
            \caption{An illustration of the mechanism of \textsf{DeepNC-L}. The first three steps are shown as an example.}
            \label{fig:algorithm}
    \end{center}
\end{figure*}

We propose \textsf{DeepNC-L} that approximates the optimal solution to (\ref{eq:6}) under the assumption that there are no missing edges in $G_O$, which implies that the non-existent edges between nodes in $G_O$ are regarded as observable entries in ${\bf{S}}^\pi$.
Since $\Phi$ indicates the set of edge existence probabilities and is thus obtained from the set of learned model  parameters $\Theta$ for each $\pi$, (\ref{eq:6}) can be simplified to the problem of finding a node ordering $\hat{\pi}$ such that
\begin{equation}
\label{eq:9}
\hat{\pi}=\argmax_{\pi}p({\bf \tilde{S}}^\pi;\Theta),
\end{equation} 
where $\bf \tilde{S}^{\pi}$ is the sequence after data imputation under a given $\pi$.

To efficiently solve (\ref{eq:9}), we present two judicious approximation methods in the following. First, we design a {\em greedy} strategy that selects a single node at each inference (generation) step. More precisely, instead of exhaustively searching for the node ordering that maximizes  $p({\bf \tilde{S}}^\pi;\Theta)$ among $(|V_O|+|V_M|)!$ possible permutations,  we aim to {\em consecutively} find a single node $\hat{v} \in V^{(i)}$ such that 
\begin{equation}
\begin{split}
\label{eq:11a}
&\hat{v}  = \argmax_{v\in V^{(i)}} p({\bf{\tilde{S}}}^\pi_i;{\bm \varphi}_i)\\
&\text{subject to } \pi(v) = i
\end{split}
\end{equation} 
for each step $i \in \{2,\cdots,|V_O|+|V_M|\}$, where $V^{(i)}$ is a set of nodes that have not been generated until the $i$-th inference step and $\hat{v}$ is removed from $V^{(i)}$ after each inference step. That is, $V^{(i+1)} \leftarrow V^{(i)}\char`\\ \{\hat{v}\}$ (refer to Fig.~\ref{fig:algorithm} for the node removal). We note that the first node can be arbitrarily chosen in the generation process.
Second, we further approximate the solution to (\ref{eq:11a}) by treating all unknown entries  (i.e., missing data) in $\tilde{\bf S}_i^\pi$ {\em equally} during the computation while retrieving $\hat{v}$ from the set $V^{(i)}$, rather than computing the likelihoods in (\ref{eq:11a}) along with all entries in $\tilde{\bf S}_i^\pi$. Let us define {\em two types of nodes} as observable nodes and missing nodes. Then, we select a node of either type at random in {\em proportion to the number of nodes {\em belonging to} each type} in $V^{(i)}$ to ensure that  there is no bias in the node selection. When the selected node type is ``missing'', we choose $\hat{v}$ at random from all missing nodes in $V^{(i)}$ without any computation since all missing nodes are treated equally. In contrast, when the selected node type is ``observable'', we choose an observable node based solely on the computation for the observable entries in ${\bf{S}}^\pi_{i}$ by reformulating our problem as follows:
\begin{equation}
\begin{split}
\label{eq:15}
&\hat{v}  = \argmax_{v\in V_O \cap V^{(i)}} p({\bf{O}}^\pi_i;{\bm \varphi}_i)\\
&\text{subject to } \pi(v) = i,
\end{split}
\end{equation} 
for each step $i \in \{2,\cdots,|V_O|+|V_M|\}$, where ${\bf{O}}^\pi_i$ denotes the set of observable entries in ${\bf{S}}^\pi_{i}$; $p({\bf{O}}^\pi_i;{\bm \varphi}_i) =
\prod_{{s}^\pi_{i,j} = 1}{\bm \varphi}_{i,j}\prod_{{s}^\pi_{i,j} = 0}(1-{\bm \varphi}_{i,j})$ since the observable entries are only taken into account and ${s}^\pi_{i,j}$ denotes the $j$-th element of ${\bf{S}}^\pi_i$; and $V_O \cap V^{(i)}$ indicates the set of remaining observable nodes after $i-1$ inference steps. Note that $p({\bf{O}}^\pi_i;{\bm \varphi}_i)$ is non-computable if there is no observable entry in ${\bf{S}}^\pi_{i}$.

Now, we are ready to show a stepwise description of the \textsf{DeepNC-L} algorithm. 

{\bf 1. Initialization}: For $i=1$, we set $V^{(1)}$ to $V_O \cup V_M$ and randomly choose a node in $V^{(1)}$ to be $\hat{v}$. 

{\bf 2. Node selection}: For $i \in \{2,\cdots,|V_O|+|V_M|\}$, we find $\hat{v}$ by either randomly selecting a missing node in $V^{(i)}$ or solving~(\ref{eq:15}), depending on which node type is selected.

{\bf 3. Data imputation}: After finding $\hat{v}$, we apply a {\em data imputation} strategy of the missing part (i.e., unknown entries) in ${\bf S}_i^\pi$ through the inference process of GraphRNN-S. Specifically, suppose that $\pi(u) = i$ and $\pi(v) = j$, which means that the $i$-th and $j$-th nodes in a given node ordering $\pi$ are $u$ and $v$, respectively. Then, we have
\begin{equation}
\label{eq:6c}
    \tilde{s}_{ij}^\pi= 
\begin{cases}
    \text{Bernoulli}({\bm \varphi}_i[j]),& \text{if } u \notin V_O \text{ or } v \notin V_O\\
    {s}_{i,j}^\pi,              & \text{otherwise},
\end{cases}
\end{equation}
where the Bernoulli trial with the probability ${\bm \varphi}_i[j]$ maps the value of the unknown entry to 1 if the outcome ``success" occurs and to 0 otherwise. 

{\bf 4. Repetition}: We iterate the second and third steps $|V_O|+|V_M|-1$ times until the recovered graph is fully generated. 

For a more intuitive understanding, consider the following example.

{\bf Example 1}: As illustrated in Fig.~\ref{fig:algorithm}, let us describe three steps to select the first three nodes of a given graph according to the aforementioned procedure. We start by randomly assigning the first node of the inference process to node M$_1$ (i.e., $\pi(\text{M}_1) = 1$ and $V^{(2)} \leftarrow V^{(1)} \char`\\ \{\text{M}_1\}$). Since we do not have any information about the connections for the unseen node M$_1$, ${{s}}^\pi_{2,1}$ is unknown for all nodes $v \in V^{(2)}$. Suppose that we generate an observable node at this step by random selection. Since there is no observable entry in ${\bf{S}}^\pi_{i}$, we randomly choose node A among the three nodes in $V_O \cap V^{(2)}$ as the second node and set $\pi(\text{A}) = 2$, resulting in $V^{(3)}\leftarrow V^{(2)} \setminus \{\text{A} \}$. Assuming that ${\bm \varphi_2} = [0.9]$ and a Bernoulli trial with the probability ${\bm \varphi_2}$ returns $1$, we impute ${\tilde{s}}^\pi_{2,1}$ with $1$ according to (\ref{eq:6c}). Let us turn to the next step in order to select the third node. In this case, since nodes B and C belong to the  type of observable nodes, ${\tilde{s}}^\pi_{3,2}$ takes the value of either 1 or 0, depending on the connections to node A. Suppose that we again generate an observable node at this step and ${\bm \varphi_3} = [0.75, 0.2]$. When either $\pi(\text{B})=3$ or $\pi(\text{C})=3$, the likelihood $p({\bf{O}}^\pi_3;{\bm \varphi_3})$ can be computed as:
\begin{itemize}
\item If $\pi(\text{B}) = 3$, then it follows that $p({\bf{O}}^\pi_3;{\bm \varphi_3}) = {\bm \varphi_{3,2}} = 0.2$ using (\ref{eq:4b}).
\item If $\pi(\text{C}) = 3$, then it follows that $p({\bf{O}}^\pi_3;{\bm \varphi_3}) = 1-{\bm \varphi_{3,2}} = 1-0.2 = 0.8$  in a similar manner.
\end{itemize}
Based on the above results, setting $\pi(\text{C})$ to 3 leads to the maximum value of $p({\bf{O}}^\pi_3;{\bm \varphi_3})$, which is thus the solution to the problem in (\ref{eq:15}) for $i=3$. As depicted in Fig.~\ref{fig:algorithm}, node C is chosen in this step. By assuming that a Bernoulli trial with the probability ${\bm \varphi_{3,1}}=0.75$ returns $1$, we finally have ${\bf \tilde{S}}^\pi_3=[1,0]$.

\subsubsection{Computational Efficiency}

In the following, we examine how to efficiently compute the likelihoods in (\ref{eq:15}) through a complexity reduction technique.
We start by making a helpful observation as illustrated in Fig.~\ref{fig:example2}. Suppose that nodes M$_1$, A, B, and E from the original graph with 8 observable nodes and 3 missing nodes have already been generated sequentially  after four inference steps. Then, one can see that ${\bf{O}}^\pi_5 = 0$ when node D, G, or H is selected in the fifth step (i.e., $\pi(\text{D})=5$, $\pi(\text{G})=5$, or $\pi(\text{H})=5$) since each of the three nodes has no connection to the nodes A, B, and E that have already been generated. Consequently, the likelihood $p({\bf{O}}^\pi_5;{\bm \varphi}_5)$ is identical for these three cases. We generalize this observation in the following lemma.
\begin{lemma}
\label{lem:1}
Let $L^{(i)}$ denote the set of not yet selected direct neighbors of observable nodes generated for $i-1$ inference steps, expressed as
\begin{equation}
\label{eq:16}
    L^{(i)}= 
\begin{cases}
   (L^{(i-1)} \cup \mathcal{N}(\hat{v}))\cap V^{(i)}, & \text{\upshape if } \hat{v} \in V_O\\
   L^{(i-1)} \cap V^{(i)},              & \text{\upshape otherwise},
\end{cases}
\end{equation}
where $i \in \{2,\cdots,|V_O|+|V_M|\}$, $L^{(1)} = \varnothing$, $\hat{v}$ is the selected node in the $(i-1)$-th step, and $\mathcal{N}(\hat{v})$ is the set of (direct) neighbors of $\hat{v}$.
Then, the likelihood $p({\bf{O}}^\pi_i;{\bm \varphi}_i)$ in (\ref{eq:15}) is the same for all $u \notin L^{(i)}$, where $u \in V_O$ and $\pi(u)=i$. 
\end{lemma}
\begin{proof}
For the observable node $u$ that does not belong to the set $L^{(i)}$ and is not generated for $i-1$ inference steps, all observable entries in ${\bf S}_i^\pi$ (i.e., entries in ${\bf O}_i^\pi$) take the value of 0's since there is no associated edge. Thus, it follows that $p({\bf{O}}^\pi_i;{\bm \varphi}_i) = \prod_{{s}^\pi_{i,j} = 0}(1-{\bm \varphi}_{i,j})$, which is identical for all $u \notin L^{(i)}$, where $u \in V_O$ and $\pi(u)=i$. This completes the proof of this lemma.
\end{proof}
Lemma~\ref{lem:1} allows us to compute the likelihood $p({\bf{O}}^\pi_i;{\bm \varphi}_i)$ only once for all nonselected observable nodes $u \notin L^{(i)}$ when solving (\ref{eq:15}), which corresponds to the case where node D, G, or H is selected in the fifth step in Fig.~\ref{fig:example2} while $L^{(5)}=\{\text{C,F}\}$, indicating the set of nonselected neighbors of nodes A, B, and E.

\begin{figure}[t]
    \begin{center}
            \includegraphics[width=0.75\linewidth]{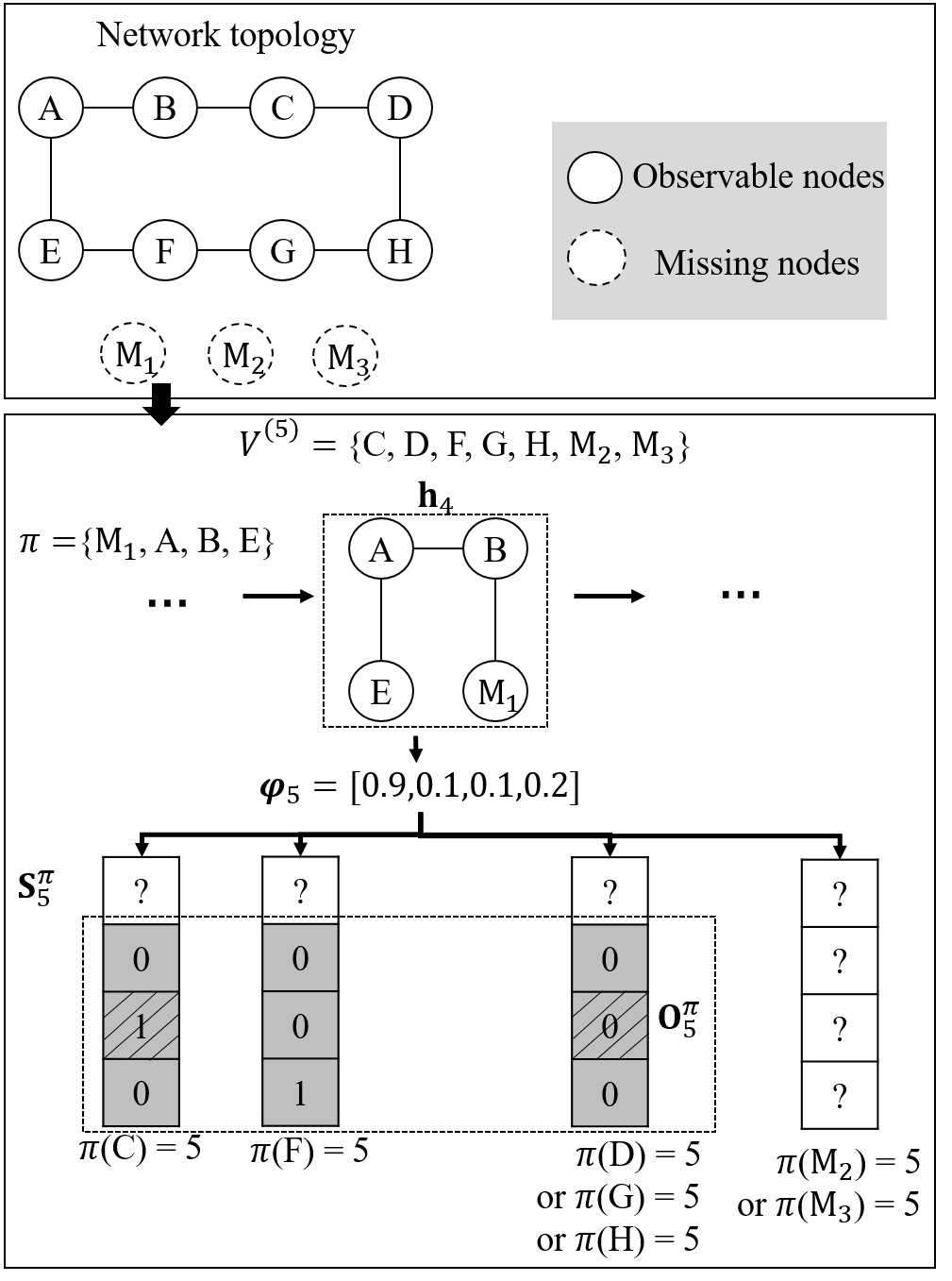}
            \caption{An example illustrating the fifth inference step of \textsf{DeepNC-L}, where nodes M$_1$, A, B, and E have  been generated sequentially.}
            \label{fig:example2}
    \end{center}
\end{figure}

Next, we explain how to efficiently solve the problem in (\ref{eq:15}) without computing likelihoods $p({\bf{O}}^\pi_i;{\bm \varphi}_i)$ for observable nodes. From Fig.~\ref{fig:example2}, one can see that $s_{5,3}^\pi$ (corresponding to entries with diagonal lines in ${\bf O}_5^\pi$) is the only term that makes the difference between two sets ${\bf{O}}^\pi_5$ for the cases when node  C is selected and when either node D, G, or H is selected, which implies that it may not be necessary to compute the likelihoods of $s_{5,2}^\pi$ and $s_{5,4}^\pi$ for node selection. Thus,  from the fact that most of the entries in ${\bf{O}}^\pi_i$ tend to be 0's in many real-world networks that are usually sparse,  the computational complexity can be greatly reduced if we make the comparison  of likelihoods in (\ref{eq:15}) based only on the entries  in ${\bf{O}}^\pi_i$ that have a value of 1. To this end, we eliminate all terms $(1-{\bm \varphi}_{i,j})$ corresponding to ${s}_{i,j}^\pi=0$ from $p({\bf{O}}^\pi_i;{\bm \varphi}_i)$ when a node $v \in V_O \cap V^{(i)}$ is selected. For computational convenience, we define
\begin{equation}
\begin{split}
\label{eq:distance}
D_v &= \frac{\prod_{{s}_{i,j}^\pi=1}{\bm \varphi}_{i,j}\prod_{{s}_{i,j}^\pi=0}(1-{\bm \varphi}_{i,j})}
{\prod_{{s}_{i,j}^\pi \in {\bf{O}}^\pi_i}{(1-{\bm \varphi}_{i,j})}} \\
&= \prod_{{s}_{i,j}^\pi=1}\frac{{\bm \varphi}_{i,j}}{(1-{\bm \varphi}_{i,j})}
\end{split}
\end{equation} 
for $v\in V_O \cap V^{(i)}$. Since the denominator in (\ref{eq:distance}) is the same for all $v \in V_O \cap V^{(i)}$, it is obvious that $\hat{v}=\argmax_v{D_v}$ is the solution to (\ref{eq:15}). We note that computing $D_v$ is less computationally  expensive than computing $p({\bf{O}}^\pi_i;{\bm \varphi}_i)$ when the number of entries with the value of 1's in ${\bf O}_i^\pi$ is low. As a special case in which all observable entries in ${\bf S}_i^\pi$ take the value of 0's, the denominator in (\ref{eq:distance}) is equivalent to $p({\bf{O}}^\pi_i;{\bm \varphi}_i)$, from which it follows that $D_u = 1$ due to the fact that a node $u\notin L^{(i)}$ is selected. Thus, if $D_v < 1$ for all $v \in L^{(i)}$, then the likelihood  in (\ref{eq:15}) for selecting a node $u \notin L^{(i)}$ is higher than that for selecting a node $v\in L^{(i)}$. In this case, we randomly choose a node $\hat{v} \notin L^{(i)}$ without further computation based on Lemma~\ref{lem:1}. In consequence, we compute $D_v$ only for nodes in the set $L^{(i)}$, rather than computing $D_v$ for all nodes in $V_O \cap V^{(i)}$. The following example describes how the computational complexity can be  reduced according to the aforementioned technique by revisiting Fig.~\ref{fig:example2}.

{\bf Example 2}: Suppose that we generate an observable node at the fifth inference step and ${\bm \varphi}_5 = [0.9, 0.1, 0.1, 0.2]$. In this step, one can see that $L^{(5)}= \{\text{C},\text{F}\}$; thus, instead of computing the likelihood $p({\bf{O}}^\pi_3;{\bm \varphi_5})$ in (\ref{eq:15}) five times for all  nonselected observable nodes C, D, F, G, and H in $V^{(5)}$, we only compute $D_\text{C} =\frac{{\bm \varphi}_{5,3}}{1-{\bm\varphi}_{5,3}}= \frac{0.1}{1-0.1}$ and $D_\text{F}=\frac{{\bm\varphi}_{5,4}}{1-{\bm\varphi}_{5,4}}= \frac{0.2}{1-0.2}$ from (\ref{eq:distance}). Since both $D_\text{C}$ and $D_\text{F}$ are smaller than 1, we randomly choose one of the three observable nodes D, G, and H that are not in $L^{(5)}$ as $\hat{v}$.


\subsubsection{Stepwise Summary of \textsf{DeepNC-L}}\label{sec:413}

\begin{algorithm}[t]
\DontPrintSemicolon
\KwIn{${G_O}, |V_M|, f_\text{out}, f_\text{trans}$}
\KwOut{$(\hat{\pi}, \Phi)$}
\textbf{Initialization}: $ i \leftarrow 2; {\bf h}_0 \leftarrow$ random initialization; ${\bf \tilde{S}}_1^{\pi} \leftarrow \varnothing;  \hat{v} \leftarrow v \in V_O \cup V_M;\pi(\hat{v})\leftarrow 1;$\qquad
$L^{(1)}\leftarrow \varnothing;$ Update $L^{(i)}$ according to (\ref{eq:16});\\
\SetKwBlock{Begin}{function}{end function}
\Begin(\textsf{DeepNC-L})
{
  \While {$i \geq |V_O|+|V_M|$}  {
	${\bf h}_{i-1}  \leftarrow f_\text{trans}({\bf h}_{i-2}, {\bf\tilde{S}}^{\pi}_{i-1})$\;
	${\bm \varphi}_{i}  \leftarrow f_\text{out}({\bf h}_{i-1})$\;
	Select a node type\;
	\uIf {the selected node type is ``observable''}{
		\For{$ v \in L^{(i)}$} {
		Compute $D_v$ according to (\ref{eq:distance})\;
		}
		\uIf {\upshape ($D_v < 1$ for all $v$ or $L^{(i)}=\varnothing$) and $L^{(i)} \neq V_O \cap V^{(i)}$}  {
			Randomly select an observable node $\hat{v} \notin L^{(i)}$\;
		}
		\uElse{
			$\hat{v}\leftarrow \argmax_v D_v$\;
		}
		Update $L^{(i)}$ according to (\ref{eq:16})\;
    	}
    	\uElse{
		Randomly select an unobservable node $\hat{v}$\;
      }
     ${\bf\tilde{S}}^\pi_i \leftarrow$ Impute ${\bf S}_i^\pi$ according to (\ref{eq:6c})\;
	$\pi(\hat{v}) \leftarrow i+1$\;
	$i \leftarrow i + 1$\;
   }
  \Return{$(\hat{\pi},\Phi)$}
}
\caption{\textsf{DeepNC-L}}\label{al1}
\end{algorithm}

We summarize the overall procedure of our \textsf{DeepNC-L} algorithm in Algorithm 1.
We initially select the first node at random, and then start the inference process by identifying connections for the next node according to the following four stages: 

{\bf 1.} Using the two functions $f_\text{trans}$ and $f_\text{out}$ in  (\ref{eq:5a}) and (\ref{eq:5}), respectively, we obtain ${\bm \varphi}_{i}$ (refer to lines 4--5). 

{\bf 2.} Let $m$ denote the cardinality of the set of missing nodes that can be potentially generated in the $i$-th step. We then randomly select a node type so that the selected node is missing with probability of $\frac{m}{|V_O|+|V_M|-i+1}$ (refer to line 6). 

{\bf 3a.} If the type of observable nodes is selected, then we compute $D_v$, which is a function of $\varphi_i$, according to (\ref{eq:distance}) for all $v \in L^{(i)}$. When $D_v < 1$ for all $v \in L^{(i)}$ or $L^{(i)}=\varnothing$, we randomly select an observable node $\hat{v} \notin L^{(i)}$ provided that $L^{(i)} \neq V_O \cap V^{(i)}$. Otherwise, we select the node $\hat{v}$ that maximizes $D_v$. Afterwards, we update $L^{(i)}$ by including neighbors of the selected node $\hat{v}$ (refer to lines 7--14).

{\bf 3b.} If the type of missing nodes is selected, then we select one node $\hat{v}$ randomly among all missing nodes that have not been generated until the $i$-th step. (refer to lines 15--16).

{\bf 4.} The data imputation process takes place before the next iteration  of node generation. Finally, we update the node ordering $\pi$ by including the selected node $\hat{v}$ for the $i$-th step. The algorithm continues by repeating stages 1--4 and terminates when a fully inferred sequence ${\bf{S}}^\pi$ is generated (refer to lines 17--20). 

We remark that a node ordering $\hat{\pi}$ is found given a set of edge existence probabilities $\Phi$, which is inferred by our model parameters $\Theta$, while assuming that $G_O$ is a complete subgraph; thus, the resulting tuple $(\hat{\pi},\Phi)$ may not be accurate when there are missing edges in $G_O$. This motivates us to develop the \textsf{DeepNC-EM} algorithm in the following subsection.

\subsection{\textsf{DeepNC-EM} Algorithm}\label{sec:4a2}

In this subsection, we introduce \textsf{DeepNC-EM} to further improve the performance of \textsf{DeepNC-L} by relaxing the assumption that there are no missing edges between two nodes in $G_O$. A na\"ive recovery of $G_O$ even with state-of-the-art link prediction methods before conducting network completion may lead to suboptimal performance since the network structures of $G_O$ are potentially distorted due to the effect of missing nodes and missing incident edges. Thus, we aim to find the most likely configuration of three types of missing edges in the set $E_M$ specified in Section~\ref{sec:3a1} by {\em jointly} estimating a tuple $(\pi,\Phi)$. To this end, we solve (\ref{eq:6}) by designing an improved \textsf{DeepNC} method using the EM algorithm.

We now describe the proposed \textsf{DeepNC-EM}, which is built upon the  \textsf{DeepNC-L}  algorithm in Section~\ref{sec:4a}. Let 
$(\pi^{(0)}, \Phi^{(0)})$  and $Z$ denote the initial output of \textsf{DeepNC-L}  and the set of non-existent edges between nodes in $G_O$, respectively. First, we estimate the {\em potential existence} likelihoods of edges in $Z$, denoted by $\Phi_Z$, by extracting $|V_O|^2-E_O$ elements corresponding to $Z$ from the likelihoods $\Phi^{(0)}$ of all edges under the node ordering $\pi^{(0)}$. Then, the E-step samples $Z^{(t)}$ from $p(Z^{(t)}|\Phi_Z^{(t)})$ via Bernoulli trials to create multiple instances of $G_O^{(t)}$, where the supercript $(t)$ denotes the EM iteration index. In the M-step, we adopt \textsf{DeepNC-L} to subsequently optimize the parameters $\Phi_Z$ given the samples obtained in the E-step. The EM iteration alternates between performing the E-step and M-step according to the following expressions, respectively:

{\em E-step:} ${Z^{(t)} \sim p(Z|\Phi_Z^{(t)})},$

{\em M-step:} $\Phi_Z^{(t+1)}=\argmax_{\Phi_Z} \mathbb{E}[p(Z^{(t)}|\Phi_Z)].$

The overall procedure of \textsf{DeepNC-EM} is summarized in Algorithm~\ref{al2}. Here,  Filter($\pi^{(t)}[i], \Phi^{(t)}[i]$) in lines 1 and 10 is invoked to retrieve $\Phi_Z^{(t)}$ from $\Phi^{(t)}$; $\eta >0$ is an arbitrarily small threshold indicating a stopping criterion for the algorithm; $\Delta_s$ denotes the number of samples in each E-step; and $[i]$ indicates the sample index.

\begin{algorithm}[t]
\DontPrintSemicolon
\KwIn{$\pi^{(0)}, \Phi^{(0)}, {G_O}, |V_M|, f_\text{out}, f_\text{trans}, \Delta_s$}
\KwOut{$(\hat{\pi}, \hat{\Phi})$}
\textbf{Initialization}: $t \leftarrow 0; \Phi_Z^{(0)} \leftarrow \text{Filter}(\pi^{(0)}, \Phi^{(0)});$
\SetKwBlock{Begin}{function}{end function}
\Begin(\textsf{DeepNC-EM})
{
  \Do {$\norm{\Phi_Z^{(t)}-\Phi_Z^{(t-1)}}_2 < \eta$}  {
	E-step:\;
		\For{$i \in \{1,\cdots,\Delta_s\}$} {
		       $Z^{(t)}[i] \sim p(Z|\Phi_Z^{(t)})$\;
			$G_O^{(t)}[i]  \leftarrow \text{ add edges sampled from } Z^{(t)}[i]$\;
		}
	M-step:\;
     		\For{$i \in \{1,\cdots,\Delta_s\}$} {
			 $(\pi^{(t+1)}[i], \Phi^{(t+1)}[i]) \leftarrow$ \textsf{DeepNC-L}(${G_O^{(t)}[i]}, |V_M|, f_\text{out}, f_\text{trans}$)\;
			$\Phi_Z^{(t+1)}[i] \leftarrow \text{Filter}(\pi^{(t+1)}[i], \Phi^{(t+1)}[i])$\;
		}
        
	$\Phi_Z^{(t+1)} \leftarrow \frac{1}{\Delta_s}\sum_i \Phi_Z^{(t+1)}[i]$\;	
	$t \leftarrow t + 1$\;
   }
   $\hat{Z} \sim p(Z|\Phi_Z^{(t+1)})$\;
   $\hat{G}_O \leftarrow \text{ add edges from }\hat{Z}$\;
   $(\hat{\pi},\hat{\Phi}) \leftarrow$ \textsf{DeepNC-L}(${\hat{G}_O}, |V_M|, f_\text{out}, f_\text{trans}$)\;
  \Return{$(\hat{\pi},\hat{\Phi})$}
}
\caption{\textsf{DeepNC-EM}}\label{al2}
\end{algorithm}

\subsection{Complexity Analysis}\label{sec:4b}
In this subsection, we analyze the computational complexities of the \textsf{DeepNC-L} and \textsf{DeepNC-EM}  algorithms.  
\subsubsection{Complexity of \textsf{DeepNC-L}}\label{sec:4b1}

We start by examining the complexity of each inference step $i \in \{2,\cdots,|V_O|+|V_M|\}$. It is not difficult to show that the complexity is dominated by the case in which an observable node is selected in the inference process. Note that it is possible to compute $D_v$ in constant time as the average degree over a network is typically regarded as a constant~\cite{degreeconstant}. Thus, the complexity of this step is bounded by $\mathcal{O}(|L^{(i)}|)$ since we exhaustively compute $D_v$ over the nodes $v \in L^{(i)}$. The data imputation process is computable in constant time when parallelization can be applied since the Bernoulli trials are independent of each other. As our algorithm is composed of $|V_O|+|V_M|-1$ inference steps, the total complexity is finally given by $\mathcal{O}((|V_O|+|V_M|)|L^{(i)}|)$, which can be rewritten as $\mathcal{O}(|V_O|\cdot|L^{(i)}|)$ from to the fact that $|V_M| \ll |V_O|$. The following theorem states a comprehensive analysis of this computational complexity.

\begin{theorem}
Lower and upper bounds on the computational complexity of the proposed \textsf{DeepNC-L} algorithm are given by $\Omega(|V_O|)$  and $\mathcal{O}(|V_O|^2)$, respectively.
\end{theorem}
\begin{proof}
The parameter $L^{(i)}$ is the set of neighboring nodes to the observable nodes that have already been generated in the $i$-th step, while its cardinality depends on the network topology. For the best case where all nodes are isolated with no neighbors, we always have $|L^{(i)}|=0$ for each generation step; thus, each step is computable in constant time, yielding the total complexity of $\Omega(|V_O|)$.
For the worst case, corresponding to a fully-connected graph, it follows that $|L^{(i)}|=|V_O|+|V_M|-i$ for each generation step, thus yielding the total complexity of  $\mathcal{O}(|V_O|^2)$. This completes the proof of this theorem.
\end{proof}

From Theorem 1, it is possible to establish the following corollary.
\begin{corollary}
The computational complexity of the \textsf{DeepNC-L} algorithm scales as $\Theta(|V_O|^{1+\epsilon})$, where $0\le\epsilon\le1$ depends on a given network topology, i.e., the sparsity of networks.
\end{corollary}

We shall validate this assertion in Corollary 1 via empirical evaluation for various datasets in the next section by identifying that $\epsilon$ is indeed small, which implies that the complexity of \textsf{DeepNC-L} is almost linear in $|V_O|$.

\subsubsection{Complexity of \textsf{DeepNC-EM}}\label{sec:4b2}
We turn to examining the computational complexity of each EM step to finally analyze the overall complexity. In the E-step, we can parallelize both the Bernoulli trials for edge sampling and the operation that adds sampled edges to $G_O^{(t)}[i]$ in lines 5 and 6, respectively. Consequently, the computational complexity of each E-step is given by $\mathcal{O}(\Delta_s)$, where $\Delta_s$ is the number of samples in each E-step. The M-step is dominated by \textsf{DeepNC-L} as the function Filter$(\cdot,\cdot)$ can also be executed in parallel since all operations therein are performed independently of each other. Thus, the computational complexity of each M-step is given by $\mathcal{O}(\Delta_s |V_O|^{1+\epsilon})$. When the number of EM iterations  is given by $k_{\text{EM}}$, determined by the threshold $\eta$, and there are a total of $\Delta_s$ samples, the complexity of \textsf{DeepNC-EM} is finally given as $\Theta(k_{\text{EM}}\Delta_s|V_O|^{1+\epsilon})$ based on Corollary 1. Since both $k_\text{EM}$ and $\Delta_s$ are regarded as constants as in~\cite{kronem}, the total computational complexity scales as $\Theta(|V_O|^{1+\epsilon})$.

\section{Experimental Evaluation}\label{sec:5}

In this section, we first describe both synthetic and real-world datasets that we use in the evaluation. We also present three state-of-the-art methods for network completion as a comparison. After presenting a performance metric and our experimental settings, we intensively evaluate the performance of our \textsf{DeepNC} algorithms.

\subsection{Datasets}\label{sec:5a}
Two synthetic and three real-world datasets across various domains (e.g., social, citations, and biological networks) are used as a series of homogeneous networks (graphs), denoted by $G_I$, and described in sequence. For all experiments, we treat graphs as undirected and only consider the largest connected component without isolated nodes. The statistics of each dataset, including the number of similar graphs and the range of the number of nodes, is described in Table~\ref{tab:dataset}. In the following, we summarize important characteristics of the datasets.

{\bf Lancichinetti-Fortunato-Radicchi (LFR)}~\cite{lfr}. We generate synthetic graphs with the LFR model in which the degree exponent of a power-law distribution, the average degree, the minimum community size, the community size exponent, and the mixing parameter are set to 3, 5, 20, 1.5, and 0.1, respectively. Refer to the original paper~\cite{lfr} for a detailed description of these parameters.

{\bf Barabasi-Albert (B-A)}~\cite{data_ba}. We generate further synthetic graphs using the B-A model. The attachment parameter of the model is set in such a way that each newly added node is connected to four existing nodes, unless otherwise stated.

{\bf Protein}~\cite{protein}. The protein structure network is a biological network. Each protein is represented by a graph, in which nodes represent amino acids. Two nodes are connected
if they are less than 6 Angstroms apart.

{\bf Ego-CiteSeer}~\cite{citeseer}. This CiteSeer dataset is an online citation network and is a frequently used benchmark. Nodes and edges represent publications and citations, respectively.

{\bf Ego-Facebook}~\cite{data_facebook}. This Facebook dataset is a social friendship network extracted from Facebook. Nodes and edges represent people and friendship ties, respectively.

\begin{table}[t]
\caption{Statistics of 5 datasets, where NG and NN denote the number of similar graphs and the range of the number of nodes in each dataset, respectively, including training graphs $G_I$ and a test graph $G_T$. Here, k denotes $10^3$.}
\label{tab:dataset}
\centering
\begin{tabular}{|l|l|l|}
\hline
Name                                             & NG      & NN       \\ \hline
LFR & 500  & 1.6k--2k                                                                                 \\ \hline
B-A                   & 500  & 1.6k--2k                                                                       \\ \hline
Protein                  & 918  & 100--500                                                                         \\ \hline
Ego-CiteSeer                                        & 737   & 50--399    \\ \hline
Ego-Facebook                                          & 10  & 52--1,034  \\ \hline
\end{tabular}
\end{table}

\subsection{State-of-the-art Approaches}\label{sec:4b}

In this subsection, we present three state-of-the-art network completion approaches for comparison.

{\bf KronEM}~\cite{kronem}. This approach aims to infer the missing part of a true network based solely on the connectivity patterns in the observed part via a generative graph model based on Kronecker graphs, where the parameters are estimated via an EM algorithm.

{\bf EvoGraph}~\cite{evograph}. To solve the network completion problem, EvoGraph infers the missing nodes and edges in such a way that the topological properties of the observable network are preserved via an efficient preferential attachment mechanism.

{\bf A variant of GraphRNN-S}. As a na\"ive approach for network completion using deep generative models of graphs, we modify the inference process of the original GraphRNN-S~\cite{graphRNN} so that it can be used as a network completion method as follows. Under a random ordering of  observable nodes, we first obtain the sequence  \{${\bf S}^{\pi}_2,\cdots,{\bf S}^{\pi}_{|V_O|}$\} along with the observable entries from $G_O$. Then,  by invoking the inference process of GraphRNN-S, we generate $|V_M|$ missing nodes using trained functions $f_\text{trans}$ and $f_\text{out}$ based on \{${\bf S}^{\pi}_2,\cdots,{\bf S}^{\pi}_{|V_O|}$\}. This variant of GraphRNN-S for network completion is termed vGraphRNN in our evaluation.

\subsection{Performance Metric}\label{sec:5b}

To assess the performance of our proposed method and the above state-of-the-art approaches, we need to quantify the degree of agreement between the recovered graph and the original graph. To this end, we adopt the GED as a well-known performance metric.

\begin{definition}
\label{def:2}
{\bf Graph edit distance (GED)}~\cite{ged}. Given a set of graph edit operations, the GED between a recovered graph $\hat{G}$ and the true graph ${G}$  is defined as
\begin{equation} 
\text{\upshape GED}(\hat{G},{G})=\min_{(e_{1},...,e_{k})\in {\mathcal {P}}(\hat{G},{G})}\sum_{i=1}^{k}c(e_{i}),
\end{equation} 
where ${\mathcal {P}}(\hat{G},{G})$ denotes the set of edit paths transforming $\hat{G}$ into a graph isomorphic to ${G}$, and $c(e)\geq 0$ is the cost of each graph edit operation $e$. 
\end{definition}

Note that only four operations are allowed in our setup, including vertex substitution, edge insertion, edge deletion, and edge substitution, and $c(e)$ is identically set to 1 for all operations. Since the problem of computing the GED is NP-complete~\cite{gednphard}, we adopt an efficient approximation algorithm proposed in~\cite{approxgednew}. In our experiments, GED is normalized by the average size of the two graphs.

\subsection{Experimental Setup}\label{sec:setup}
We first describe the settings of the neural networks. In our experiments, the function $f_\text{trans}$ is implemented by using 4 layers of GRU cells with a 128-dimensional hidden state; and the function $f_\text{out}$ is implemented by using a two-layer perceptron with a 64-dimensional hidden state and a sigmoid activation function. The Adam optimizer~\cite{adamopt} is used for minibatch training with a learning rate of 0.001, where each minibatch contains 32 graph sequences. We train the model for 32,000 batches in all experiments. 

To test the performance of our method, we randomly select one graph from each dataset to act as the underlying true network $G_T$. From each dataset, we select all remaining similar graphs as training data $G_I$  unless otherwise stated.

To create a partially observable network from the true network $G_T$, we adopt the following two graph sampling strategies from~\cite{graphsampling}. The first strategy, called {\em random node} (RN) sampling, selects nodes uniformly at random to create a sample graph. The second strategy, {\em forest fire} (FF) sampling, starts by picking a seed node uniformly at random and adding it to a sample graph (referred to as burning). Then, FF sampling burns a fraction of the outgoing links with nodes attached to them. This process is repeated recursively for each neighbor that is burned until no new node is selected to be burned. Afterwards, we sample uniformly at random a portion of edges from the complete subgraph sampled from $G_T$ to finally acquire $G_O$.  In our experiments, the partially observable network $G_O$ is constructed by 90\% of edges in a complete subgraph consisting of 70\% of nodes sampled from $G_T$ unless otherwise specified. Each experimental result is averaged over 10 executions.

\subsection{Experimental Results}\label{sec:5c}
Our empirical study in this subsection is designed to answer the following five key research questions.
\begin{itemize}
\item {\em Q1}. How much does the performance of \textsf{DeepNC-EM} improve with respect to the number of EM iterations?
\item {\em Q2}. How much do the \textsf{DeepNC} algorithms improve the performance of network completion over the state-of-the-art approaches?
\item {\em Q3}. How beneficial are the \textsf{DeepNC} algorithms in more difficult situations where either a large number of nodes and edges are missing or the training data are also incomplete? 
\item {\em Q4}. How robust is \textsf{DeepNC-EM} to the portion of missing edges in $G_O$ in comparison with the other state-of-the-art approaches?
\item {\em Q5}. How scalable are \textsf{DeepNC} algorithms with the size of the graph?
\end{itemize}

To answer these questions, we carry out six comprehensive experiments as follows.

\subsubsection{Comparative Study Between \textsf{DeepNC-L} and \textsf{DeepNC-EM} (Q1)}
In Fig.~\ref{fig:exp1}, we show the performance of the \textsf{DeepNC-EM} algorithm proposed in Section~\ref{sec:4a2} with respect to GED according to the number of EM iterations using two synthetic datasets, i.e., the LFR and B-A models. As shown in Fig.~\ref{fig:exp1}, our findings are as follows:
\begin{itemize}
\item For both RN and FF sampling strategies, the GED of \textsf{DeepNC-EM} decreases as the number of EM iterations increases.
\item The number of EM iterations required to achieve a sufficiently low GED value is relatively small compared to the network size. This can be seen from the LFR dataset, where the performance improvement is marginal after four iterations.
\item We observe that \textsf{DeepNC-EM} exhibits less fluctuations over EM iterations on the LFR dataset. This might be caused by the fact that graphs generated using the LFR model are denser than those using the B-A model under our  setting, which enables the algorithm to be more likely  to correctly recover the edges connecting two nodes in the set $V_O$.
\end{itemize}
In the subsequent experiments, the number of EM iterations is set to 6.

\begin{figure}[t]
    \begin{center}
\hspace{-0.2cm}
	 \begin{subfigure}[]{0.485\columnwidth}
		\centering
            \includegraphics[height=3.5cm]{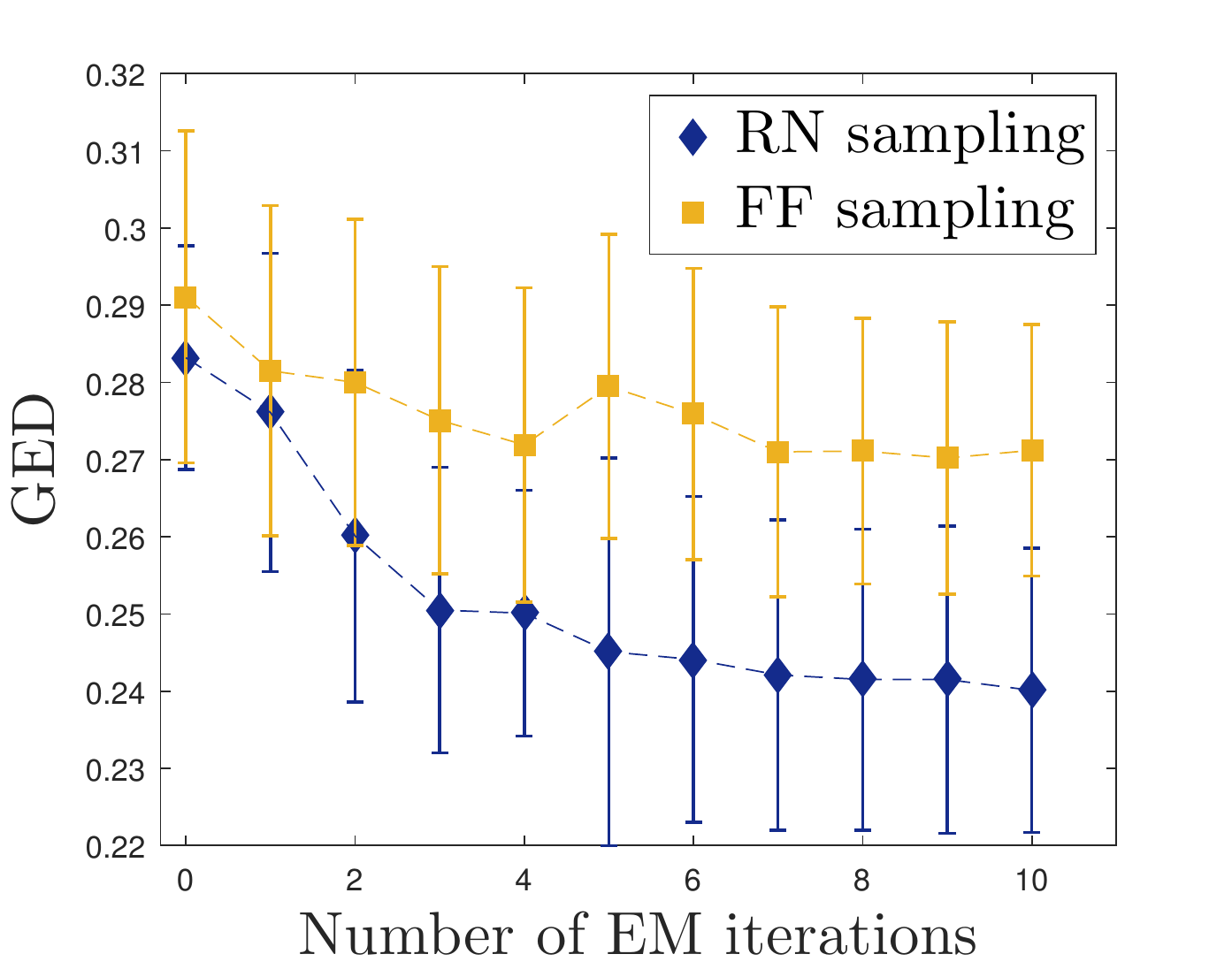}
            \caption{The LFR dataset}
            \label{fig:fig5a}
        \end{subfigure}
\hspace{-0.05cm}
\begin{subfigure}[]{0.485\columnwidth}
		\centering
            \includegraphics[height=3.5cm]{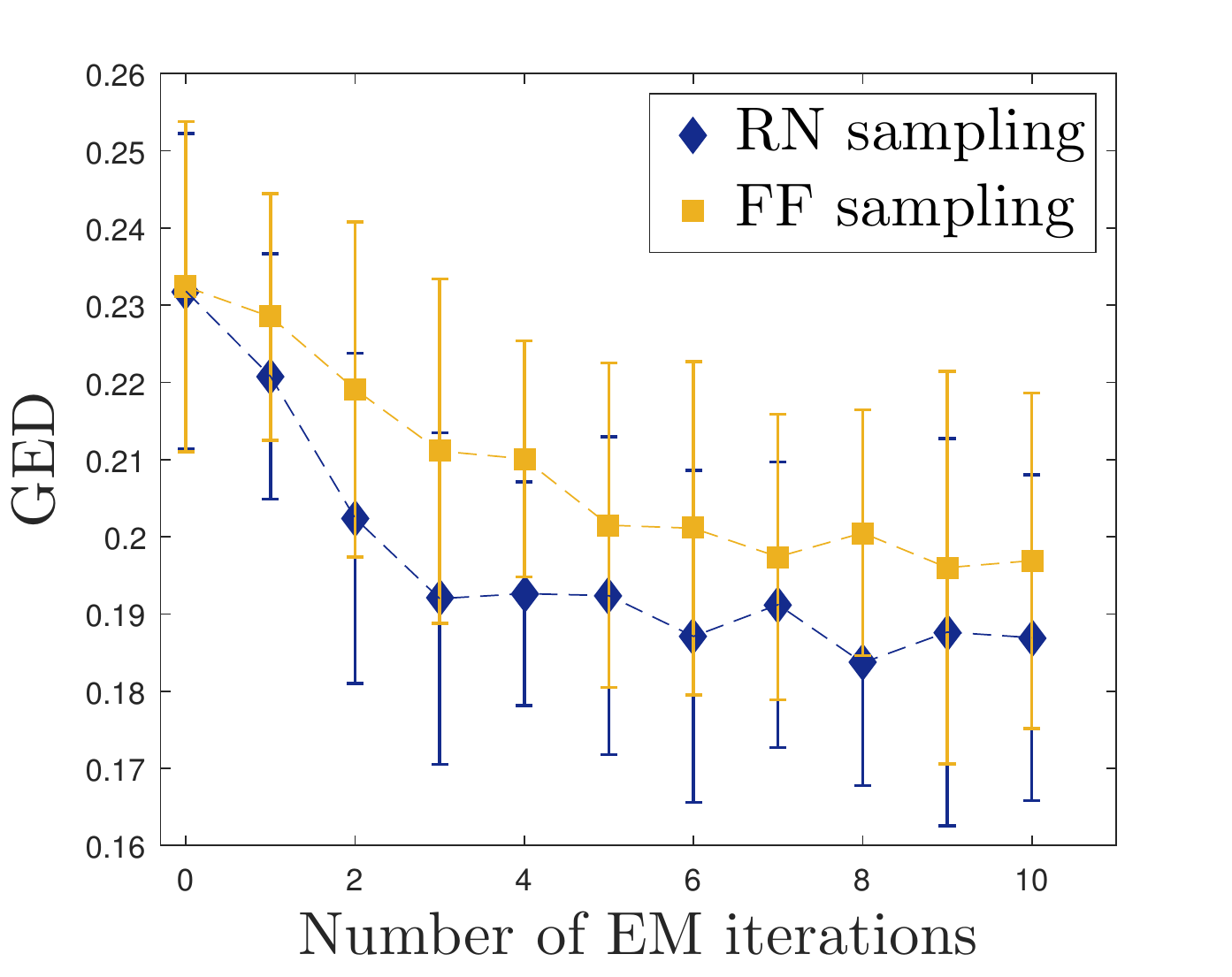}
            \caption{The B-A dataset}
            \label{fig:fig5b}
        \end{subfigure}
            \caption{GED of \textsf{DeepNC-EM} over the number of EM iterations. Here, the performance of \textsf{DeepNC-L} corresponds to the case where the number of EM iterations is zero. }
            \label{fig:exp1}
    \end{center}
\end{figure}


\begin{table*}[!h]
\centering
\caption{Performance comparison in terms of graph edit distance (average $\pm$ standard deviation). Here, the best method for each dataset is highlighted using bold fonts.}
\label{tab:resultGED}
\bgroup
\def\arraystretch{1.2}
\resizebox{\linewidth}{!}{%
\begin{tabular}{|l|l|l|l|l|l|l|l|l|l|l|}
\hline
\multirow{2}{*}{\backslashbox{Dataset}{Method}} & \multicolumn{1}{c|}{\multirow{2}{*}{\shortstack{{\bf DeepNC-EM}  \\ ($X$)}}} & \multicolumn{1}{c|}{\multirow{2}{*}{\shortstack{{\bf DeepNC-L}}}} &  \multicolumn{1}{c|}{\multirow{2}{*}{\shortstack{\textbf{vGraphRNN} \\ ($Y_1$)}}} & \multicolumn{1}{c|}{\multirow{2}{*}{\shortstack{\textbf{KronEM} \\ ($Y_2$)}}} & \multicolumn{1}{c|}{\multirow{2}{*}{\shortstack{\textbf{EvoGraph}\\ ($Y_3$)}}} & \multicolumn{3}{c|}{\textbf{Gain (\%)}}                                                                   
\\ \cline{7-9}
\multicolumn{1}{|c|}{} & \multicolumn{1}{c|}{} &  \multicolumn{1}{c|}{} & \multicolumn{1}{c|}{} & \multicolumn{1}{c|}{}  & \multicolumn{1}{c|}{} &\multicolumn{1}{c|}{$\frac{Y_1-X}{Y_1}\times100$} &\multicolumn{1}{c|}{$\frac{Y_2-X}{Y_2}\times100$}  &\multicolumn{1}{c|}{$\frac{Y_3-X}{Y_3}\times100$}
\\
\hline
LFR (RN)                                            
& {\bf 0.2793}  $\pm$ 0.0145  & 0.2864  $\pm$ 0.0206  & 0.3099 $\pm$ 0.0241 & 0.3713 $\pm$ 0.0428 & 0.5126 $\pm$ 0.0124 & 9.87 & 24.78  & 45.51\\
LFR (FF)                                           
& {\bf 0.2612}  $\pm$ 0.0205  & 0.2801  $\pm$ 0.0214  & 0.3155 $\pm$ 0.0197 & 0.3671 $\pm$ 0.0278 & 0.4512 $\pm$ 0.0075 & 17.21 & 28.85  & 42.11
\\
B-A (RN)                                            
& {\bf 0.1782}  $\pm$ 0.0120  & 0.1888  $\pm$ 0.0104  & 0.2015 $\pm$ 0.0210 & 0.3921 $\pm$ 0.0304 & 0.5612 $\pm$ 0.0084 & 11.56 & 54.55  & 68.25
\\
B-A (FF) 
& {\bf 0.1811}  $\pm$ 0.0106  & 0.2024  $\pm$ 0.0134  & 0.2041 $\pm$ 0.0202 & 0.3706 $\pm$ 0.0418 & 0.5455 $\pm$ 0.0087 & 11.27 & 51.13  & 66.80
\\
Protein (RN)                                            
& {\bf 0.2616}  $\pm$ 0.0521  & 0.3015  $\pm$ 0.0520  & 0.3861 $\pm$ 0.2101 & 0.4565 $\pm$ 0.1077 & 0.4422 $\pm$ 0.0014 & 32.25 & 42.69  & 40.84
\\
Protein (FF)                                            
& {\bf 0.2603}  $\pm$ 0.0571  & 0.3012  $\pm$ 0.0481  & 0.3761 $\pm$ 0.1121 & 0.4455 $\pm$ 0.1240 & 0.4111 $\pm$ 0.0025 & 30.79 & 41.57  & 36.68
\\
Ego-CiteSeer (RN)                                            
& {\bf 0.3012}  $\pm$ 0.0414  & 0.3236  $\pm$ 0.0414  & 0.4915 $\pm$ 0.2514 & 0.5811 $\pm$ 0.0438 & 0.9166 $\pm$ 0.0109 & 39.16 & 48.17  & 67.14
\\
Ego-CiteSeer (FF)                                            
& {\bf 0.3241}  $\pm$ 0.0571  & 0.3458  $\pm$ 0.0511  & 0.5416 $\pm$ 0.1918 & 0.5571 $\pm$ 0.0518 & 0.9013 $\pm$ 0.0041 & 40.16 & 41.82  & 64.04
\\
Ego-Facebook (RN)                                            
& {\bf 0.4213}  $\pm$ 0.0502  & 0.4535  $\pm$ 0.0508  & 0.5928 $\pm$ 0.2015 & 0.6167 $\pm$ 0.0268 & 0.8161 $\pm$ 0.0121 & 28.93 & 31.68  & 48.38  
\\
Ego-Facebook (FF)                                            
& {\bf 0.4711}  $\pm$ 0.0471  & 0.5021  $\pm$ 0.0604  & 0.6182 $\pm$ 0.1897 & 0.6160 $\pm$ 0.0447 & 0.7222 $\pm$ 0.0104 & 23.79 & 23.52  & 34.77
\\
\hline
\end{tabular}
}
\egroup
\end{table*}
\subsubsection{Comparison With State-of-the-Art Approaches (Q2)}
The performance comparison between two \textsf{DeepNC} algorithms and three state-of-the-art network completion methods, including vGraphRNN, KronEM~\cite{kronem}, and EvoGraph~\cite{evograph}, with respect to GED is presented in Table~\ref{tab:resultGED} for all five datasets. We note that \textsf{DeepNC-EM}, \textsf{DeepNC-L}, and vGraphRNN use structurally similar graphs as training data $G_I$; meanwhile, both KronEM and EvoGraph operate based solely on the partially observable graph $G_O$  without any training phase.
We observe the following:

\begin{itemize}
\item The improvement rates of \textsf{DeepNC-EM} over vGraphRNN, KronEM, and EvoGraph are up to 40.16\%, 54.55\%, and 68.25\%, respectively. These maximum gains are achieved for the Ego-CiteSeer and B-A datasets.
\item The \textsf{DeepNC-L} and \textsf{DeepNC-EM} algorithms are insensitive to sampling strategies for creating a partially observable network, whereas the performance of EvoGraph depends on the sampling strategy. Specifically, sampling via FF results in better performance than that via RN sampling when EvoGraph is used due to the fact that the FF sampling strategy tends to preserve the network properties such as the degree distribution~\cite{graphsampling}. In reality, if the sampling strategy is unknown and one only acquires randomly sampled data, then graph upscaling methods such as EvoGraph would certainly perform poorly. 
This result displays the robustness of our \textsf{DeepNC} algorithms to graph samplings.
\item Even with deletions of only 10\% of edges, the additional gain of \textsf{DeepNC-EM} over \textsf{DeepNC-L} is still significant for all datasets. The maximum improvement rate of 13.58\% is achieved on the Protein dataset.
\item In a comparison of the performance differences between KronEM and EvoGraph, KronEM performs better in most cases. However, KronEM is inferior to EvoGraph in the case where the degree distribution of a network does not strictly follow the pure power-law degree distribution. EvoGraph consistently outperforms KronEM in the Protein dataset. 
\item The standard deviation of GED is relatively high when vGraphRNN is employed (e.g., 0.2514 for the Ego-CiteSeer dataset), which demonstrates that a random node ordering of observable nodes for network completion does not guarantee a stable solution.
\end{itemize}
Consequently, \textsf{DeepNC-EM} consistently outperforms all state-of-the-art methods for all synthetic and real-world datasets, which reveals the robustness of our method toward diverse network topologies.

\begin{table*}[t]
\centering
\caption{Performance comparison in terms of graph edit distance when 70\% of nodes are missing (average $\pm$ standard deviation). Here, the best method for each dataset is highlighted using bold fonts.}
\label{tab:resultGED2}
\bgroup
\def\arraystretch{1.2}
\resizebox{\linewidth}{!}{%
\begin{tabular}{|l|l|l|l|l|l|l|l|l|l|l|}
\hline
\multirow{2}{*}{\backslashbox{Dataset}{Method}} & \multicolumn{1}{c|}{\multirow{2}{*}{\shortstack{{\bf DeepNC-EM}  \\ ($X$)}}} & \multicolumn{1}{c|}{\multirow{2}{*}{\shortstack{{\bf DeepNC-L}}}} &  \multicolumn{1}{c|}{\multirow{2}{*}{\shortstack{\textbf{vGraphRNN} \\ ($Y_1$)}}} & \multicolumn{1}{c|}{\multirow{2}{*}{\shortstack{\textbf{KronEM} \\ ($Y_2$)}}} & \multicolumn{1}{c|}{\multirow{2}{*}{\shortstack{\textbf{EvoGraph}\\ ($Y_3$)}}} & \multicolumn{3}{c|}{\textbf{Gain (\%)}}                                                                   
\\ \cline{7-9}
\multicolumn{1}{|c|}{} & \multicolumn{1}{c|}{} &  \multicolumn{1}{c|}{} & \multicolumn{1}{c|}{} & \multicolumn{1}{c|}{}  & \multicolumn{1}{c|}{} &\multicolumn{1}{c|}{$\frac{Y_1-X}{Y_1}\times100$} &\multicolumn{1}{c|}{$\frac{Y_2-X}{Y_2}\times100$}  &\multicolumn{1}{c|}{$\frac{Y_3-X}{Y_3}\times100$}
\\
\hline
LFR                                            
& {\bf 0.2902}  $\pm$ 0.1204  & 0.3251  $\pm$ 0.1245  & 0.3516 $\pm$ 0.1284 & 0.6167 $\pm$ 0.0802 & 0.7177 $\pm$ 0.0212 & 17.46 & 52.94  & 59.57
\\
B-A                                            
& {\bf 0.2611}  $\pm$ 0.1021  & 0.2635  $\pm$ 0.1018  & 0.2644 $\pm$ 0.1487 & 0.6547 $\pm$ 0.0728 & 0.8273 $\pm$ 0.0140 & 1.25 & 60.12  & 68.44
\\
Protein                                          
& {\bf 0.3244}  $\pm$ 0.1014  & 0.3648  $\pm$ 0.1189  & 0.4678 $\pm$ 0.2428 & 0.9674 $\pm$ 0.0437 & 0.7272 $\pm$ 0.0161 & 30.65 & 66.47  & 55.39
\\
Ego-CiteSeer                                           
& {\bf 0.3414}  $\pm$ 0.1144  & 0.3988  $\pm$ 0.1171  & 0.6031 $\pm$ 0.3125 & 0.7727 $\pm$ 0.0578 & 0.9161 $\pm$ 0.0116 & 43.39 & 55.82  & 62.73
\\
Ego-Facebook                                        
& {\bf 0.5685}  $\pm$ 0.1412  & 0.5875  $\pm$ 0.1280  & 0.6448 $\pm$ 0.2985 & 0.8027 $\pm$ 0.0689 & 0.9505 $\pm$ 0.1057 & 11.83 & 29.18  & 40.19  
\\
\hline
\end{tabular}
}
\egroup
\end{table*}

\subsubsection{Applicability to Fringe Scenarios (Q3)}
\label{sec:553}
We now compare our \textsf{DeepNC} algorithms to the three state-of-the-art network completion methods in more difficult settings that often occur in real environments: 1) the case in which a large portion of nodes are missing and 2) the case in which training graphs are also only partially observed. In these experiments, we only show the results for the RN sampling strategy since the results from FF sampling follow similar trends.

First, we create a partially observable network $G_O$ consisting of only 30\% of nodes from the underlying true graph ${G}_T$ via sampling. The performance comparison between the \textsf{DeepNC} algorithms and the three state-of-the-art methods with respect to GED is presented in Table~\ref{tab:resultGED2} for all five datasets. As shown in Tables \ref{tab:resultGED} and \ref{tab:resultGED2}, a large number of missing nodes and edges result in significant performance degradation for KronEM and EvoGraph, while \textsf{DeepNC-EM}, \textsf{DeepNC-L}, and vGraphRNN are more robust as the latter three methods take advantage of the topological information from similar graphs  (i.e., training data) to infer the missing part.

Next, we perform RN sampling so that only a part of nodes in the training graphs is observable. In Fig.~\ref{fig:partially}, we compare the GED of the two \textsf{DeepNC} algorithms and the three state-of-the-art methods, where the degree of observability in training graphs is set to $\{95,90\}$\%.
We find that the \textsf{DeepNC} algorithms still outperform the state-of-the-art methods on all datasets with the exception of the Ego-Facebook dataset, where the performance of \textsf{DeepNC-L} is slightly inferior to that of KronEM when 90\% of nodes in training graphs are observable. 

\begin{figure}[t]

	 \begin{subfigure}[]{\columnwidth}
\hspace{-0.5cm}
            \includegraphics[width=10.2cm]{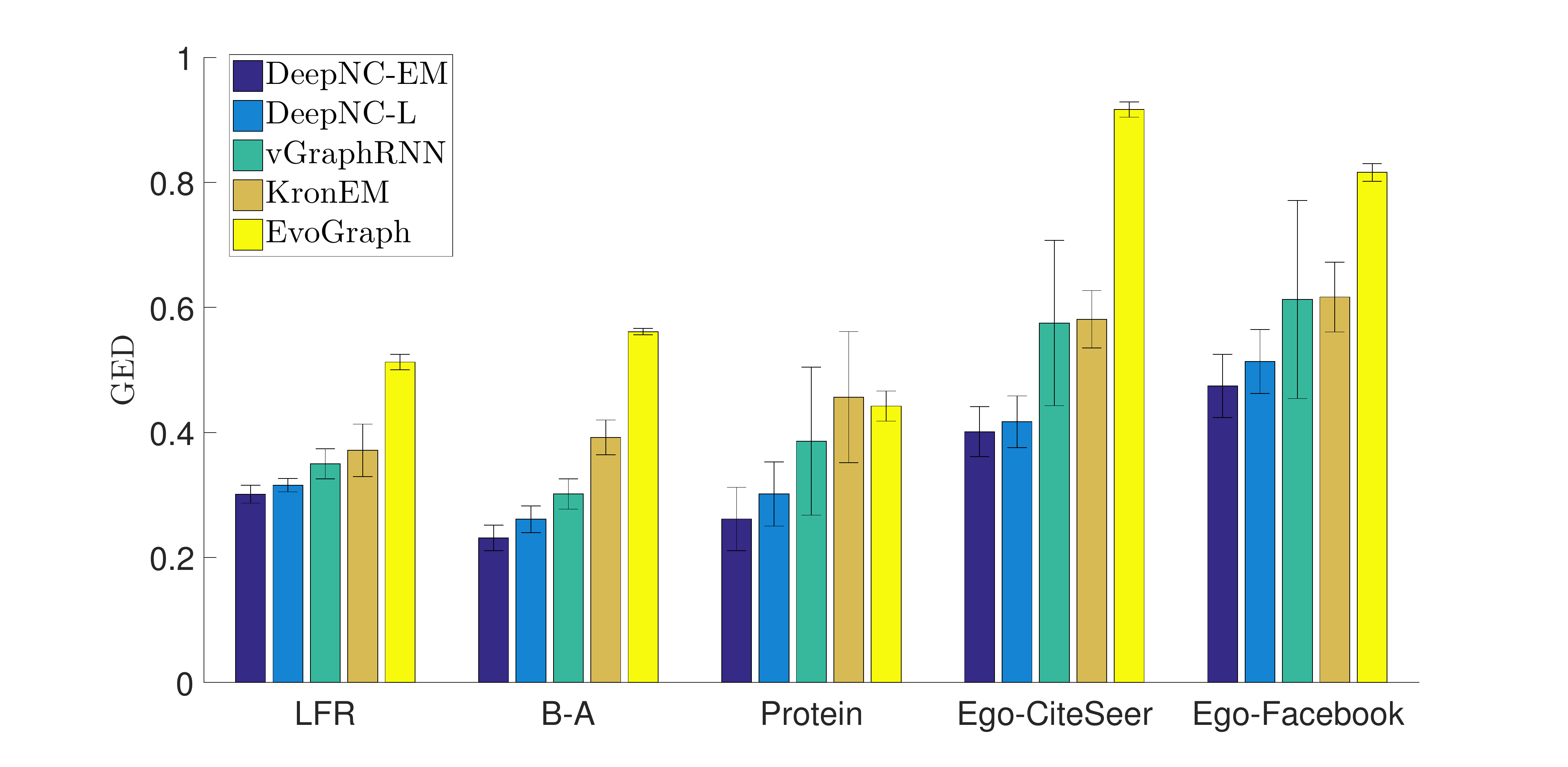}
            \caption{95\% observability}
            \label{fig:fig8a}
        \end{subfigure}
\begin{subfigure}[]{\columnwidth}
\hspace{-0.5cm}
            \includegraphics[width=10.2cm]{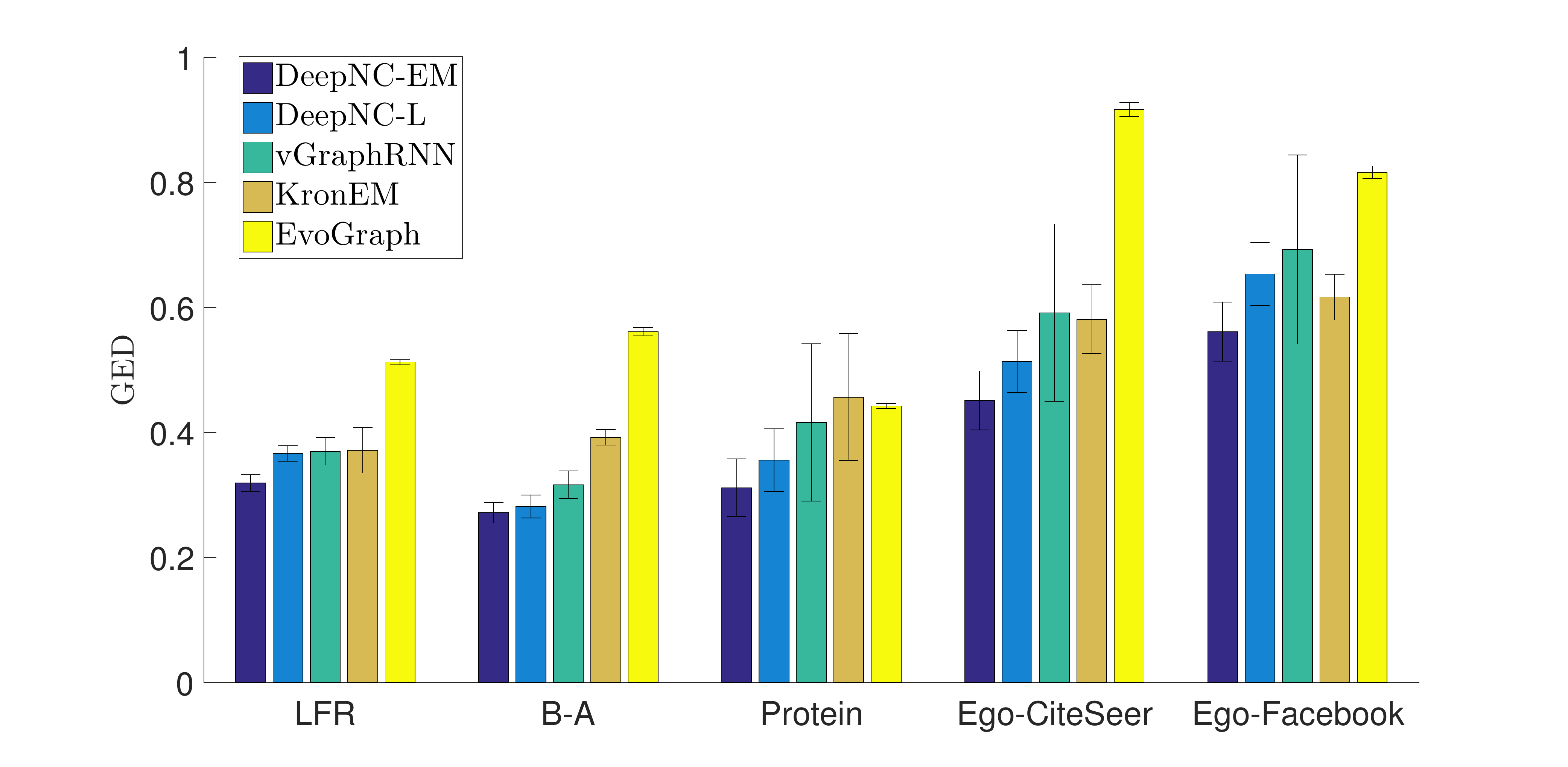}
            \caption{90\% observability}
            \label{fig:fig8b}
        \end{subfigure}
\caption{Performance comparison in terms of GED (the lower the better), where the degree of observability in training graphs is set to $\{95, 90\}$\%.}
\label{fig:partially}
\end{figure}

\subsubsection{Robustness  to the Degree of Edge Observability in $G_O$ (Q4)}
We evaluate the GED performance in the second fringe scenario, in which a partially observable network $G_O$ is created by deleting a large portion of edges uniformly at random from a complete subgraph that consists of 70\% of nodes sampled from $G_T$. In Fig.~\ref{fig:missingportion}, the performance of the \textsf{DeepNC} algorithms is compared to the state-of-the-art network completion methods using two synthetic datasets, where the fraction of missing edges is set to $\{10,15,20\}\%$. Our main findings are: 1) \textsf{DeepNC-L} outperforms the three state-of-the-art methods in all cases; 2) the gain of \textsf{DeepNC-EM} over \textsf{DeepNC-L} is more substantial when the LFR dataset is used since missing edges are inferred more accurately; and 3) both \textsf{DeepNC} algorithms exhibit less performance degradation as the number of missing edges increases, which demonstrates the robustness of our method for various degrees of edge observability. 

\begin{figure}[t]
\centering
	 \begin{subfigure}[]{\columnwidth}
		\centering
            \includegraphics[width=8cm]{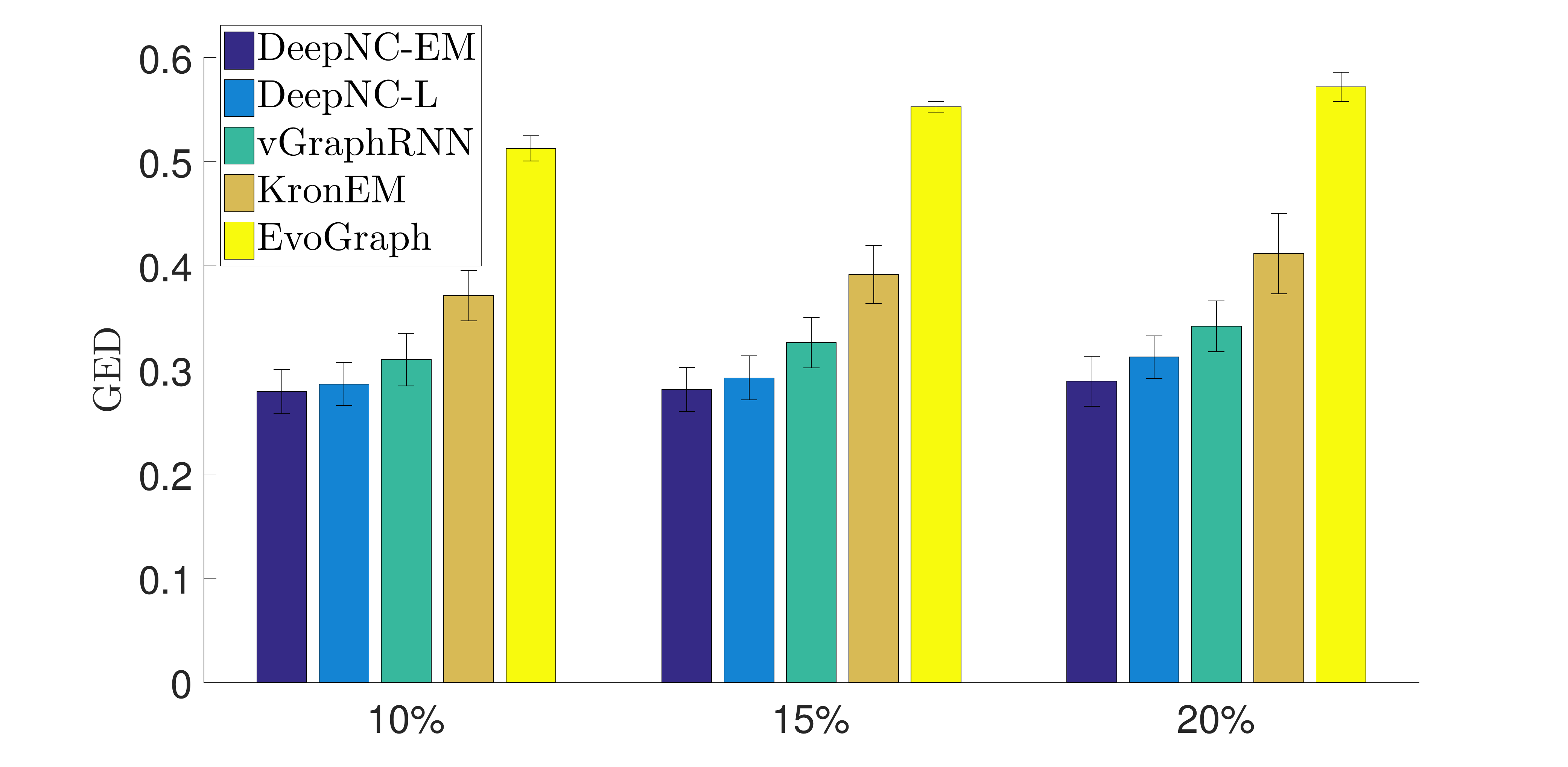}
            \caption{The LFR dataset}
            \label{fig:fig9a}
        \end{subfigure}
\hspace{-0.05cm}
\begin{subfigure}[]{\columnwidth}
		\centering
            \includegraphics[width=8cm]{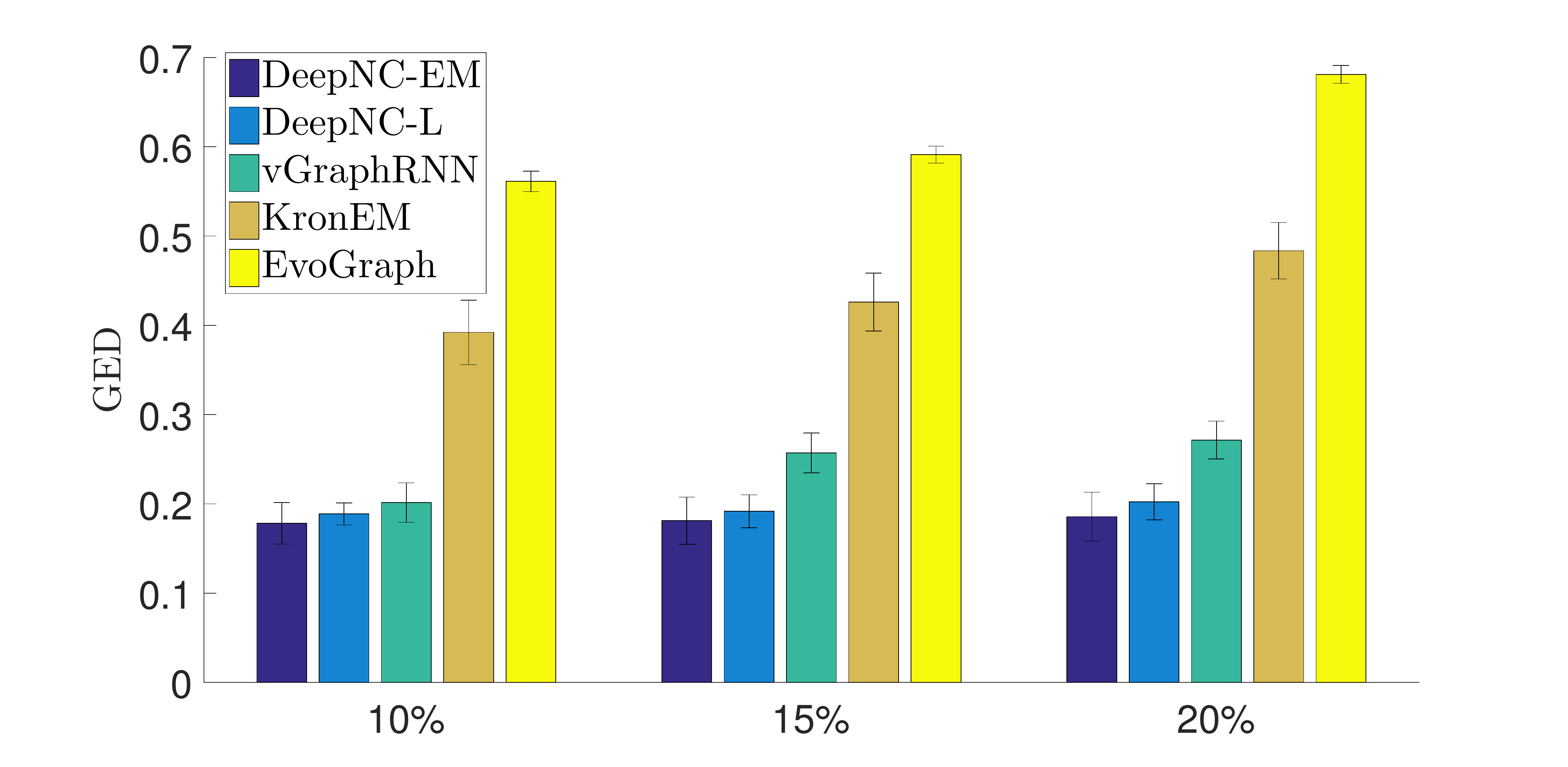}
            \caption{The B-A dataset}
            \label{fig:fig9b}
        \end{subfigure}
\caption{Performance comparison in terms of GED (the lower the better), where the degree of missingness in edges between nodes in $G_O$ is set to $\{10,15,20\}\%$.}
\label{fig:missingportion}
\end{figure}

From Tables \ref{tab:resultGED}--\ref{tab:resultGED2} and Figs.~\ref{fig:partially}--\ref{fig:missingportion}, it is worth noting that the proposed \textsf{DeepNC-EM} algorithm outperforms all state-of-the-art methods for all types of datasets under various fringe scenarios and experimental settings.

\subsubsection{Scalability (Q5)}
Finally, we empirically show the average runtime complexity via experiments using the three sets of B-A synthetic graphs, which can conveniently be scaled up while preserving the same structural properties, where the number of connections from each new node to existing nodes is set to $c \in \{2, 4, 8\}$. In these experiments, we focus on evaluating the complexity of \textsf{DeepNC-EM} since EM iterations take constant time by executing \textsf{DeepNC-L} for each iteration. In each set of graphs, the number of  nodes, $|V_O|+|V_M|$, varies from 200 to 2,000 in increments of 200; and 30\% of nodes and their associated edges are deleted by RN sampling to create partially observable networks. Other parameter settings follow those in Section~\ref{sec:setup}. In Fig.~\ref{fig:complexity}, we illustrate the log-log plot of the execution time in seconds versus $|V_O|$, where each point represents the average complexity over 10 executions of \textsf{DeepNC-EM}. In the figure, dotted lines are also shown from the analytical result  with a proper bias, showing a tendency that  slopes of the lines for $c\in\{2,4,8\}$ are approximately given by 1.16, 1.26, and 1.41, respectively. This indicates that the computational complexity of \textsf{DeepNC-EM} is dependent on the average degree in a given graph. Moreover, we find that an almost linear complexity in $|V_O|$, i.e., $\Theta(|V_O|^{1+\epsilon})$ for a small $\epsilon>0$, is attainable since the slopes are at most 1.41 even for the relatively dense graph corresponding to $c=8$.

\begin{figure}[t]
\centering
\includegraphics[width=8cm]{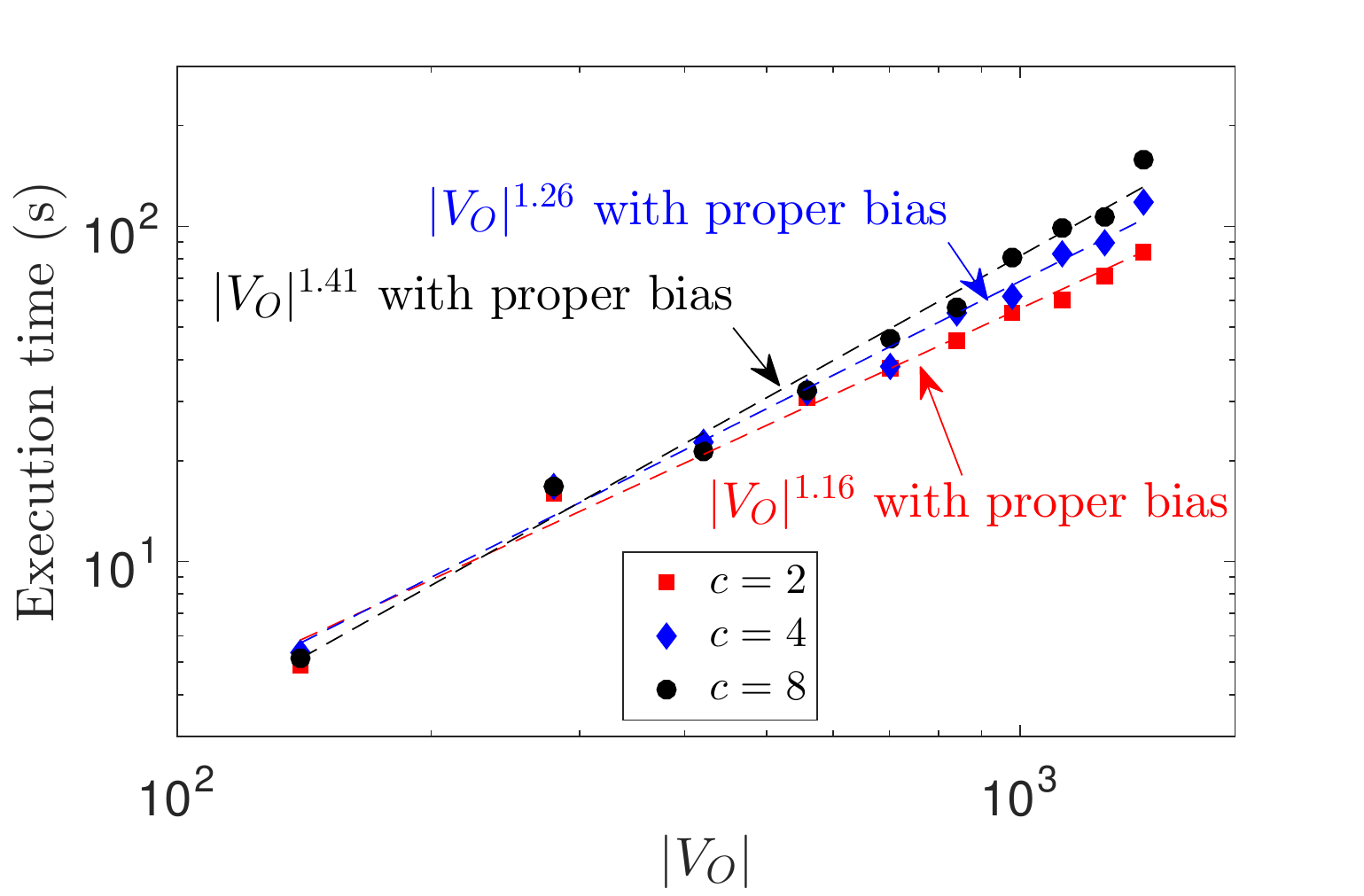}
\caption{The computational complexity of \textsf{DeepNC-EM}, where the log-log plot of the execution time versus $|V_O|$ is shown.}
\label{fig:complexity}
\end{figure}

\section{Concluding Remarks}\label{sec:6}

In this paper, we explored the open problem of recovering not only missing edges between observable nodes but also entirely hidden nodes and associated edges of an underlying true network. To tackle this new challenge, we introduced a novel method, termed \textsf{DeepNC}, that infers such missing nodes and edges via deep learning. Specifically, we presented an approach to first learning a likelihood over edges via an RNN-based generative graph model by using structurally similar graphs as training data and then inferring the missing parts of the network by applying an imputation strategy that restores the missing data. Furthermore, we proposed two \textsf{DeepNC} algorithms whose runtime complexities are almost linear in $|V_O|$. Using various synthetic and real-world datasets, we demonstrated that our \textsf{DeepNC} algorithms not only remarkably outperform vGraphRNN, KronEM, and EvoGraph methods, but are also robust to many difficult and challenging situations that often occur in real environments such as 1) a significant  portion of unobservable nodes, 2) training graphs that are only partially observable, or 3) a large fraction of missing edges between nodes in the observed network. Additionally, we analytically and empirically showed the scalability of our \textsf{DeepNC} algorithms.

Potential avenues of future research include the design of a unified framework for improving the performance of various downstream mining and learning tasks such as multi-label node classification, community detection, and influence maximization when \textsf{DeepNC} is adopted in partially observable networks. Here, the challenges lie in task-specific preprocessing that should be accompanied by network completion to guarantee satisfactory performance in each individual task.

\section*{Acknowledgments}
This research was supported by the Republic of Korea’s MSIT (Ministry of Science and ICT), under the High-Potential Individuals Global Training Program (No. 2020-0-01463) supervised by the IITP (Institute of Information and Communications Technology Planning Evaluation), by a grant of the Korea Health Technology R\&D Project through the Korea Health Industry Development Institute (KHIDI), funded by the Ministry of Health \& Welfare, Republic of Korea (HI20C0127), and by the Yonsei University Research Fund of 2020 (2020-22-0101).

\bibliography{TPAMIBib} 

\begin{thebibliography}{10}
\providecommand{\url}[1]{#1}
\csname url@samestyle\endcsname
\providecommand{\newblock}{\relax}
\providecommand{\bibinfo}[2]{#2}
\providecommand{\BIBentrySTDinterwordspacing}{\spaceskip=0pt\relax}
\providecommand{\BIBentryALTinterwordstretchfactor}{4}
\providecommand{\BIBentryALTinterwordspacing}{\spaceskip=\fontdimen2\font plus
\BIBentryALTinterwordstretchfactor\fontdimen3\font minus
  \fontdimen4\font\relax}
\providecommand{\BIBforeignlanguage}[2]{{%
\expandafter\ifx\csname l@#1\endcsname\relax
\typeout{** WARNING: IEEEtran.bst: No hyphenation pattern has been}%
\typeout{** loaded for the language `#1'. Using the pattern for}%
\typeout{** the default language instead.}%
\else
\language=\csname l@#1\endcsname
\fi
#2}}
\providecommand{\BIBdecl}{\relax}
\BIBdecl

\bibitem{kossinets2006effects}
G.~Kossinets, ``Effects of missing data in social networks,'' \emph{Soc.
  Netw.}, vol.~28, no.~3, pp. 247--268, Jul. 2006.

\bibitem{acquisti2015privacy}
A.~Acquisti, L.~Brandimarte, and G.~Loewenstein, ``Privacy and human behavior
  in the age of information,'' \emph{Science}, vol. 347, no. 6221, pp.
  509--514, Jan. 2015.

\bibitem{dey2012facebook}
R.~Dey, Z.~Jelveh, and K.~Ross, ``Facebook users have become much more private:
  A large-scale study,'' in \emph{Proc. IEEE Int. Conf. Pervasive Comput.
  Commun. Worksh.}, Lugano, Switzerland, Mar. 2012, pp. 346--352.

\bibitem{koskinen2013bayesian}
J.~H. Koskinen, G.~L. Robins, P.~Wang, and P.~E. Pattison, ``Bayesian analysis
  for partially observed network data, missing ties, attributes and actors,''
  \emph{Soc. Netw.}, vol.~35, no.~4, pp. 514--527, Oct. 2013.

\bibitem{kronem}
M.~Kim and J.~Leskovec, ``The network completion problem: {I}nferring missing
  nodes and edges in networks,'' in \emph{Proc. 2011 SIAM Int. Conf. Data
  Mining (SDM '11)}, Mesa, AZ, USA, Apr. 2011, pp. 47--58.

\bibitem{kromfac}
C.~Tran, W.-Y. Shin, and A.~Spitz, ``Community detection in partially
  observable social networks,'' \emph{arXiv preprint arXiv:1801.00132}, 2017.

\bibitem{citeseer}
P.~Sen, G.~Namata, M.~Bilgic, L.~Getoor, B.~Galligher, and T.~Eliassi-Rad,
  ``Collective classification in network data,'' \emph{AI Magazine}, vol.~29,
  no.~3, pp. 93--106, 2008.

\bibitem{protein}
P.~D. Dobson and A.~J. Doig, ``Distinguishing enzyme structures from
  non-enzymes without alignments,'' \emph{J. Molecular Bio.}, vol. 330, no.~4,
  pp. 771--783, Jul. 2003.

\bibitem{data_facebook}
A.~L. Traud, E.~D. Kelsic, P.~J. Mucha, and M.~A. Porter, ``Comparing community
  structure to characteristics in online collegiate social networks,''
  \emph{SIAM Rev.}, vol.~53, no.~3, pp. 526--543, Aug. 2011.

\bibitem{graphRNN}
J.~You, R.~Ying, X.~Ren, W.~Hamilton, and J.~Leskovec, ``Graph{RNN}: Generating
  realistic graphs with deep auto-regressive models,'' in \emph{Proc. Int.
  Conf. Machine Learning (ICML '18)}, Stockholm, Sweden, Jul. 2018, pp.
  5694--5703.

\bibitem{netgan}
A.~Bojchevski, O.~Shchur, D.~Z{\"u}gner, and S.~G{\"u}nnemann, ``Net{GAN}:
  Generating graphs via random walks,'' in \emph{Proc. Int. Conf. Machine
  Learning (ICML '18)}, Stockholm, Sweden, Jul. 2018, pp. 609--618.

\bibitem{ged}
A.~Sanfeliu and K.-S. Fu, ``A distance measure between attributed relational
  graphs for pattern recognition,'' \emph{IEEE Trans. Syst. Man Cybernetics},
  vol. SMC-13, no.~3, pp. 353--362, Jun. 1983.

\bibitem{data_com}
P.~Erdos and A.~R{\'e}nyi, ``On random graphs {I},'' \emph{Publ. Math.
  Debrecen}, vol.~6, pp. 290--297, 1959.

\bibitem{data_ba}
A.-L. Barab{\'a}si and R.~Albert, ``Emergence of scaling in random networks,''
  \emph{Science}, vol. 286, no. 5439, pp. 509--512, Oct. 1999.

\bibitem{kronfit}
J.~Leskovec, D.~Chakrabarti, J.~Kleinberg, C.~Faloutsos, and Z.~Ghahramani,
  ``Kronecker graphs: {A}n approach to modeling networks,'' \emph{J. Mach.
  Learning Res.}, vol.~11, pp. 985--1042, Feb. 2010.

\bibitem{grans}
R.~Liao, Y.~Li, Y.~Song, S.~Wang, W.~Hamilton, D.~K. Duvenaud, R.~Urtasun, and
  R.~Zemel, ``Efficient graph generation with graph recurrent attention
  networks,'' in \emph{Proc. Advances Neural Inf. Processing Syst. (NIPS '19)},
  Vancouver, Canada, Dec. 2019, pp. 4257--4267.

\bibitem{graphvae}
M.~Simonovsky and N.~Komodakis, ``Graph{VAE}: Towards generation of small
  graphs using variational autoencoders,'' in \emph{Proc. Int. Conf. Artificial
  Neural Netw. Machine Learning (ICANN '18)}, Rhodes, Greece, Oct. 2018, pp.
  412--422.

\bibitem{graphvae2}
T.~N. Kipf and M.~Welling, ``Variational graph auto-encoders,'' in \emph{NIPS
  Worksh. Bayesian Deep Learning}, Montréal, Canada, Dec. 2018.

\bibitem{gcpn}
J.~You, B.~Liu, Z.~Ying, V.~Pande, and J.~Leskovec, ``Graph convolutional
  policy network for goal-directed molecular graph generation,'' in \emph{Proc.
  Advances Neural Inf. Processing Syst. (NIPS '18)}, Montréal, Canada, Dec.
  2018, pp. 6410--6421.

\bibitem{miscgan}
D.~Zhou, L.~Zheng, J.~Xu, and J.~He, ``Misc-{GAN}: A multi-scale generative
  model for graphs,'' \emph{Fronti. Big Data}, vol.~2, pp. 3:1--3:10, Apr.
  2019.

\bibitem{GN}
Y.~Li, O.~Vinyals, C.~Dyer, R.~Pascanu, and P.~Battaglia, ``Learning deep
  generative models of graphs,'' \emph{arXiv preprint arXiv:1803.03324}, 2018.

\bibitem{lp1}
L.~L{\"u} and T.~Zhou, ``Link prediction in complex networks: A survey,''
  \emph{Phys. A: Stat. Mech. Appl.}, vol. 390, no.~6, pp. 1150--1170, Mar.
  2011.

\bibitem{linkpred2}
M.~Zhang and Y.~Chen, ``Link prediction based on graph neural networks,'' in
  \emph{Proc. Advances Neural Inf. Processing Syst. (NIPS '18)}, Montreal,
  Canada, Dec. 2018, pp. 5165--5175.

\bibitem{SVDlowrank}
P.~Jain, R.~Meka, and I.~S. Dhillon, ``Guaranteed rank minimization via
  singular value projection,'' in \emph{Proc. Advances Neural Inf. Processing
  Syst. (NIPS '10)}, Vancouver, Canada, Dec. 2010, pp. 937--945.

\bibitem{MFlowrank}
J.~P. Haldar and D.~Hernando, ``Rank-constrained solutions to linear matrix
  equations using powerfactorization,'' \emph{IEEE Signal Process. Lett.},
  vol.~16, no.~7, pp. 584--587, Jul. 2009.

\bibitem{CNNlowrank}
F.~Monti, M.~Bronstein, and X.~Bresson, ``Geometric matrix completion with
  recurrent multi-graph neural networks,'' in \emph{Proc. Advances Neural Inf.
  Processing Syst. (NIPS '17)}, Long Beach, CA, Dec. 2017, pp. 3697--3707.

\bibitem{onlinebandit1}
C.~Gentile, S.~Li, and G.~Zappella, ``Online clustering of bandits,'' in
  \emph{Proc. Int. Conf. Machine Learning (ICML '14)}, Beijing, China, Jun.
  2014, pp. 757--765.

\bibitem{onlinebandit2}
K.~Mahadik, Q.~Wu, S.~Li, and A.~Sabne, ``Fast distributed bandits for online
  recommendation systems,'' in \emph{Proc. 34th Int. Conf. Supercomput. (ICS
  '20)}, Las Vegas, NV, Jun.-Jul. 2020, pp. 1--13.

\bibitem{lowrank}
N.~Linial, E.~London, and Y.~Rabinovich, ``The geometry of graphs and some of
  its algorithmic applications,'' \emph{Combinatorica}, vol.~15, no.~2, pp.
  215--245, 1995.

\bibitem{linkMF}
A.~K. Menon and C.~Elkan, ``Link prediction via matrix factorization,'' in
  \emph{Proc. European Conf. Machine Learning Knowl. Disc. Databases (ECML PKDD
  '11)}, Athens, Greece, Sep. 2011, pp. 437--452.

\bibitem{misc}
R.~Eyal, A.~Rosenfeld, S.~Sina, and S.~Kraus, ``Predicting and identifying
  missing node information in social networks,'' \emph{ACM Trans. Knowl. Disc.
  Data}, vol.~8, no.~3, pp. 14:1--14:35, Jun. 2014.

\bibitem{sami}
S.~Sina, A.~Rosenfeld, and S.~Kraus, ``Solving the missing node problem using
  structure and attribute information,'' in \emph{Proc. 2013 IEEE/ACM Int.
  Conf. Advances Social Netw. Analysis Mining (ASONAM '13)}, Niagara Falls,
  Canada, Aug. 2013, pp. 744--751.

\bibitem{evograph}
H.~Park and M.-S. Kim, ``Evo{G}raph: {A}n effective and efficient graph
  upscaling method for preserving graph properties,'' in \emph{Proc. 24th ACM
  SIGKDD Int. Conf. Knowl. Disc. Data Mining (KDD '18)}, London, United
  Kingdom, Aug. 2018, pp. 2051--2059.

\bibitem{estimatepopulation}
T.~H. McCormick, M.~J. Salganik, and T.~Zheng, ``How many people you know?:
  {E}fficiently esimating personal network size,'' \emph{J. Am. Stat. Assoc.},
  vol. 105, no. 489, pp. 59--70, Sep. 2010.

\bibitem{GRU}
J.~Chung, C.~Gulcehre, K.~Cho, and Y.~Bengio, ``Empirical evaluation of gated
  recurrent neural networks on sequence modeling,'' in \emph{Proc. Deep
  Learning and Representation Learning Worksh.}, Montreal, Canada, Dec. 2014.

\bibitem{LSTM}
S.~Hochreiter and J.~Schmidhuber, ``Long short-term memory,'' \emph{Neural
  Comput.}, vol.~9, no.~8, pp. 1735--1780, Nov. 1997.

\bibitem{emalgorithm}
A.~P. Dempster, N.~M. Laird, and D.~B. Rubin, ``Maximum likelihood from
  incomplete data via the em algorithm,'' \emph{J. Royal Stat. Soc. Series B
  (Methodological)}, vol.~39, no.~1, pp. 1--22, 1977.

\bibitem{degreeconstant}
M.~E. Newman, ``Random graphs as models of networks,'' \emph{Proc. National
  Acad. Sci.}, vol.~99, no.~1, pp. 2566--2572, 2002.

\bibitem{lfr}
A.~Lancichinetti and S.~Fortunato, ``Benchmarks for testing community detection
  algorithms on directed and weighted graphs with overlapping communities,''
  \emph{Phys. Rev. E}, vol.~80, no.~1, pp. 016\,118:1--016\,118:8, Apr. 2009.

\bibitem{gednphard}
Z.~Zeng, A.~K.~H. Tung, J.~Wang, J.~Feng, and L.~Zhou, ``Comparing stars: On
  approximating graph edit distance,'' \emph{Proc. VLDB Endow.}, vol.~2, no.~1,
  pp. 25--36, Aug. 2009.

\bibitem{approxgednew}
A.~Fischer, K.~Riesen, and H.~Bunke, ``Improved quadratic time approximation of
  graph edit distance by combining {H}ausdorff matching and greedy
  assignment,'' \emph{Pattern Recogn. Lett.}, vol.~87, pp. 55--62, Feb. 2017.

\bibitem{adamopt}
D.~P. Kingma and J.~Ba, ``Adam: {A} method for stochastic optimization,'' in
  \emph{Proc. Int. Conf. Learning Rep. (ICLR '15)}, San Diego, CA, May 2015.

\bibitem{graphsampling}
J.~Leskovec and C.~Faloutsos, ``Sampling from large graphs,'' in \emph{Proc.
  12th ACM SIGKDD Int. Conf. Knowl. Disc. Data Mining (KDD '06)}, Philadelphia,
  PA, USA, Aug. 2006, pp. 631--636.

\end{thebibliography}
\bibliographystyle{IEEEtran}

\end{document}